\documentclass[pdflatex,sn-mathphys-num]{sn-jnl}%

\usepackage{tikz}
\usetikzlibrary{calc}
\usetikzlibrary{math}
\usepackage{graphicx}
\usepackage{mathrsfs}
\usepackage{amssymb}
\usepackage{shuffle}
\usepackage{amsmath}
\usepackage{bm}
\usepackage{amsfonts}
\usepackage{dsfont}
\usepackage{mathtools}
\usepackage{amsthm}
\usepackage{enumitem}
\usepackage{multirow}
\usepackage[title]{appendix}
\usepackage{xcolor}
\usepackage{textcomp}
\usepackage{manyfoot}
\usepackage{booktabs}
\usepackage{algorithm}
\usepackage{algorithmicx}
\usepackage{algpseudocode}
\usepackage{listings}

\newcommand{\Emmett}[5] %
{
    \draw[#1] (0, #2)
    \foreach\x in {1, ..., #4} {
        -- ++(#3, rand * #3)
    };
}
\newcommand{\bb}[1]{\mathbb{#1}}

\newcommand{\inner}[2]{\langle #1, #2 \rangle}
\DeclareMathOperator*{\esssup}{ess\,sup}
\DeclareMathOperator*{\argmin}{arg\,min}
\newcommand{\pvar}[3]{\left|\left|#1\right|\right|_{#2, #3}} %
\newcommand{\lift}[2]{#1^{(2)}_{#2}} %
\newcommand{\roughpath}[1]{\bm{#1} = (#1, #1^{(2)})} %
\newcommand{\roughholder}[4]{\left|\left|\left|\bm{#1}\right|\right|\right|_{\frac{#2}{#3}\text{-Höl}, #4}} %
\newcommand{\pvarSpaceSimple}[1]{\mathcal{V}^{#1\text{-var}}} %
\newcommand{\pvarSpace}[3]{\mathcal{V}^{#1\text{-var}}(#2, #3)} %
\newcommand{\roughMetricHolder}[4]{\varrho_{\frac{1}{#3}\text{-Höl}, #4}(\bm{#1}, \bm{#2})} %
\newcommand{\roughMetricPvar}[4]{\varrho_{#3, #4}(\bm{#1}, \bm{#2})} %
\newcommand{\ce}[2]{\mathcal{E} \left[ \varphi(#1_#2) | \mathcal{Y}_#2 \right]} %

\makeatletter
\newtheoremstyle{thmstyleone-fixed}%
  {\topsep}{\topsep}%
  {\normalfont\normalsize}%
  {}%
  {\bfseries}%
  {.}%
  { }%
  {\thmname{#1}\thmnumber{ #2}\thmnote{ (\textit{#3})}}%
\makeatother
\theoremstyle{thmstyleone-fixed}
\newtheorem{theorem}{Theorem}%
\newtheorem{assumption}{Assumption}[section]
\newtheorem{proposition}[theorem]{Proposition}%

\makeatletter
\newtheoremstyle{thmstyletwo-fixed}%
  {\topsep}{\topsep}%
  {\normalfont\normalsize}%
  {}%
  {\itshape}%
  {.}%
  { }%
  {\thmname{#1}\thmnumber{ #2}\thmnote{ (\textit{#3})}}%
\makeatother
\theoremstyle{thmstyletwo-fixed}
\newtheorem{example}{Example}%
\newtheorem{remark}{Remark}%
\newtheorem{lemma}{Lemma}[section]
\newtheorem{corLemma}{Corollary (lemma)}[lemma]

\makeatletter
\newtheoremstyle{thmstylethree-fixed}%
  {\topsep}{\topsep}%
  {\normalfont\normalsize}%
  {}%
  {\bfseries}%
  {.}%
  { }%
  {\thmname{#1}\thmnumber{ #2}\thmnote{ (\textit{#3})}}%
\makeatother
\theoremstyle{thmstylethree-fixed}
\newtheorem{definition}{Definition}%

\begin{document}

\title[Article Title]{Rough Path Approaches to Stochastic Control, Filtering, and Stopping}

\author*[1]{\fnm{Jonathan} \sur{A. Mavroforas}}\email{jonathan.a.mavroforas@alumni.uts.edu.au}

\author[2]{\fnm{Anthony} \sur{H. Dooley}}\email{anthony.dooley@uts.edu.au}

\affil[1*, 2]{\orgdiv{School of Mathematical and Physical Sciences}, \orgname{University of Technology, Sydney}}

\abstract{
\normalsize This paper presents a unified exposition of rough path methods applied to optimal control, robust filtering, and optimal stopping, addressing a notable gap in the existing literature where no single treatment covers all three areas. By bringing together key elements from Lyons’ theory of rough paths, Gubinelli’s controlled rough paths, and related developments, we recast these classical problems within a deterministic, pathwise framework. Particular emphasis is placed on providing detailed proofs and explanations where these have been absent or incomplete, culminating in a proof of the central \emph{verification theorem}, which is another key contribution of this paper. This result establishes the rigorous connection between candidate solutions to optimal control problems and the Hamilton–Jacobi–Bellman equation in the rough path setting. Alongside these contributions, we identify several theoretical challenges -- most notably, extending the verification theorem and associated results to general \emph{p}-variation with -- and outline promising directions for future research. The paper is intended as a self-contained reference for researchers seeking to apply rough path theory to decision-making problems in stochastic analysis, mathematical finance, and engineering.
}


\maketitle

\newpage

\section{Introduction}
\label{sec:introduction}

\subsection{Background}
\label{subsec:intro:background}

The theory of rough paths was initiated by Lyons in the 1990s as a deterministic framework for analyzing differential equations driven by highly irregular signals, most notably sample paths of Brownian motion. Classical stochastic calculus, particularly the Itô theory, relies on probabilistic structures and $L^2$ techniques, which obscure the pathwise nature of solutions and are restricted to semimartingale inputs. Lyons’ approach introduced an algebraic and analytic enhancement of paths via iterated integrals -- now known as geometric rough paths -- enabling a well-posed theory of integration and differential equations for low-regularity signals. This advancement combined earlier insights from Young integration and Chen’s iterated integrals, laying the foundation for a deterministic reformulation of stochastic analysis. Subsequent developments, including Gubinelli’s controlled rough paths and Hairer’s theory of regularity structures, further broadened the reach of the theory. Notably, Hairer received the Fields Medal in 2014 for his work on singular stochastic PDEs, where rough path methods played a foundational role.

These ideas naturally invite a re-examination of classical problems in stochastic analysis, particularly in control, filtering, and stopping. Each has traditionally relied on Itô calculus and semimartingale models. Yet in practice -- especially in high-frequency or irregular settings -- these assumptions become restrictive.

In stochastic control, the canonical formulation involves optimizing a cost functional for a controlled diffusion. The dynamic programming principle leads to a Hamilton–Jacobi–Bellman (HJB) equation, whose analysis depends on the Itô framework. But in systems influenced by market microstructure noise or turbulent environmental data, signals exhibit roughness incompatible with semimartingale theory. Rough paths provide a pathwise reformulation of controlled differential equations, enabling analysis of control problems driven by low-regularity signals. This opens the door to robust, model-free control strategies designed without probabilistic expectations.

A similar narrative arises in filtering, where the task is to estimate the latent state of a system from noisy observations. Classical filters, such as Kalman–Bucy and its nonlinear generalizations, characterize the conditional law of the signal via stochastic PDEs, relying fundamentally on Itô calculus. Rough path theory instead offers a pathwise framework, replacing expectations with continuous maps from observation paths to estimates. Recent advances show how rough analytic techniques yield new insights into inference for complex dynamical systems.

The problem of optimal stopping is also pathwise in nature. Classical formulations compute expectations over future trajectories and reduce to variational inequalities with free-boundary problems. But when the gain process is driven by a rough signal, such expectations lose meaning. Rough paths allow a deterministic reformulation, where stopping rules and value functions are defined as functionals of the rough path itself.

Together, these perspectives show that rough path theory does not merely extend stochastic analysis -- it reconfigures its foundations. By replacing probabilistic structures with pathwise constructions, rough paths provide a unifying framework for control, filtering, and stopping in irregular regimes, forming a powerful toolkit for modern stochastic modeling where irregularity and noise are the norm.

\subsection{Scope}
\label{subsec:intro:scope}

This paper provides a unified and rigorous account of rough path methods applied to stochastic optimal control, robust filtering, and optimal stopping, addressing a notable gap in the literature where no single work treats all three topics within a common framework. Our aim is to recast these classical problems in a deterministic, pathwise setting, demonstrating how rough path theory offers a robust foundation for decision-making under uncertainty -- particularly when classical semimartingale-based models are inadequate.

Beyond exposition, the paper contributes detailed and clarified proofs where the literature has been incomplete, together with illustrative examples that highlight key phenomena such as the degeneracy that arises when attempting to control noise directly, and how regularization of the control variation restores well-posedness. A central result -- and another key contribution -- is our proof of the \emph{verification theorem} (\textbf{Theorem~\ref{theorem:verification_theorem}}) for the case of $p$-variation with $2 \le p < 3$, which establishes the rigorous connection between candidate solutions to optimal control problems and the Hamilton–Jacobi–Bellman equation in the rough path setting.

We also identify significant theoretical challenges that remain open, including the extension of the verification theorem, dynamic programming principle, and associated formulations to general $p$-variation with $3 \le p < \infty$, as well as unresolved issues in filtering and stopping for systems driven by rough paths. These challenges form the basis for the future research directions outlined later in the paper.

The intended audience is researchers in stochastic analysis, control theory, mathematical finance, or related fields, seeking a coherent, self-contained reference on the application of rough path theory to stochastic decision problems. While some mathematical maturity is assumed, the presentation strives to remain pedagogically clear and accessible.

\subsection{A Note on Length}
\label{subsec:intro:length}

The length of this paper is deliberate. Our goal is to provide a treatment of control, filtering, and stopping in the rough path setting that is as complete as possible, with minimal reliance on omitted arguments or appeals to intuition. To this end, we include detailed proofs where the existing literature is fragmentary, alongside examples designed not merely to illustrate technical points, but to clarify the motivation for studying these problems in the first place. These examples are chosen carefully to highlight both the practical significance and the mathematical subtleties of rough path methods. The resulting exposition is longer than what might ordinarily be expected, but we believe this level of detail is essential for producing a self-contained reference that is both rigorous and transparent.

\section{Rough Path Preliminaries}
\label{sec:rough_path_preliminaries}
The purpose of this section is to equip the reader with a foundational understanding of rough path theory necessary for engaging with the subsequent sections of this paper. While we do not aim to provide a full treatment of the subject, this overview is designed to introduce key concepts and methods relevant to the discussions that follow. Readers interested in a comprehensive introduction to rough path theory are encouraged to consult detailed resources, such as \emph{A Course on Rough Paths with an Introduction to Regularity Structures} by Friz and Hairer \cite{frizHairer}, or Lyons’ foundational text, \emph{Differential Equations Driven by Rough Paths} \cite{lyons}.

It is worth noting that the notation in $\S$\ref{sec:general_integration_theory}, $\S$\ref{sec:integration_theory_for_optimal_control}, $\S$\ref{OC}, and $\S$\ref{sec:pathwise_robust_filtering} diverges from that in $\S$\ref{sec:optimal_stopping_with_signatures}. This difference arises because the latter section focuses on \( p \)-rough paths for values \( 1 \leq p < \infty \), while the former sections primarily consider \( p \in [2, 3) \). The choice of \( p \in [2, 3) \) is motivated by traditional stochastic integration theory, enabling us to analyze integration and differentiation for processes such as Brownian motion and fractional Brownian motion from a pathwise perspective. The focus on these specific values of \( p \) allows for a natural treatment of stochastic processes within the rough path framework, providing robust tools for handling the complexities of pathwise integration.

In the coming sections, we aim to clarify how rough path theory supports the integration and differentiation of such irregular processes, allowing us to bypass some of the probabilistic assumptions usually necessary in stochastic calculus. This pathwise approach provides an alternative view, particularly useful for applications where probabilistic structures are either unavailable or undesirable, and where direct handling of process paths with controlled roughness is required. The upcoming material will build upon these principles to establish a coherent framework for analysis in the specific contexts explored throughout this paper.

\subsection{General Integration Theory}
\label{sec:general_integration_theory}
The results presented in this section are primarily inspired by Gubinelli's work on controlled rough paths \cite{gubinelliControlledRoughPath}, which provides a powerful framework for integrating and differentiating certain irregular processes in a way that feels more intuitive and akin to the treatment of continuously differentiable processes. The tools we develop here are also sufficiently robust to incorporate stochastic processes such as Brownian motion and fractional Brownian motion within the rough path framework. Notably, these tools alone are sufficient to handle these processes thanks to the \emph{Lyons’ universal limit theorem} \cite{lyons}.

Throughout this paper, we denote the closed interval \([0, T]\) by \(J\), and write \(|\cdot|\) for the Euclidean norm on \(\mathbb{R}^n\). We begin by stating several standard definitions that will serve as the foundation for the results that follow.
\begin{definition}
\label{def:standard_definitions}
    Let \(X\colon J \to \mathbb{R}^n\), \(t \mapsto X_t\), and suppose \(0 \leq s \leq t \leq T\). We define the following:
    \begin{equation}
    \label{eq:simplex}
        \Delta_{[s, t]} \coloneqq \{(s, t) \in J^2 \colon 0 \leq s \leq t \leq T\}
    \end{equation}
    \begin{equation}
    \label{eq:path_difference}
        X_{s, t} \coloneqq X_t - X_s
    \end{equation}
    \begin{equation}
    \label{eq:p_variation}
        \|X\|_{p, J} \coloneqq \left[ \sup_{\mathcal{D}} \sum_{(t_i, t_{i+1}) \in \mathcal{D}} |X_{t_i, t_{i+1}}|^p \right]^{1/p}
    \end{equation}
    where \(\mathcal{D} = (t_i)_{i=0}^{n}\) denotes a partition of \(J = [0, T]\). For clarity, the supremum in \eqref{eq:p_variation} is taken over all such partitions \(\mathcal{D}\) of \(J\).
\end{definition}

\emph{Equation~\eqref{eq:simplex}} defines a simplex; \emph{equation~\eqref{eq:path_difference}} defines the increment of the path from \(s\) to \(t\); and \emph{equation~\eqref{eq:p_variation}} introduces a generalized notion of path length known as the \(p\)-variation. 

As an example, the \(1\)-variation of a continuously differentiable path \(Y_t\) coincides with its path length. This can be seen from the following computation:
\begin{equation*}
\begin{split}
    \|Y\|_{1, J} & \coloneqq \sup_{\mathcal{D}} \sum_{(t_i, t_{i+1}) \in \mathcal{D}} |Y_{t_i, t_{i+1}}|  \\
    & = \sup_{\mathcal{D}} \sum_{(t_i, t_{i+1}) \in \mathcal{D}} |Y'_{c_i}|\,|t_{i+1} - t_i| \\
    & = \int_0^T |Y'_t|\,dt,
\end{split}
\end{equation*}
where the second equality follows from the \emph{mean value theorem}.

Intuitively, the \(p\)-variation of a path suppresses small increments  --  those of magnitude less than one  --  so as to produce a finite measure of path ``length" that is sensitive to the chosen exponent \(p\).

Before formally presenting the theory, we introduce an example to illustrate the mechanics of rough integration.

\begin{example}[Mechanics of Rough Integration]
\label{example:mechanics_of_rough_integration}
    Suppose that \( f \colon \mathbb{R}^n \to \mathbb{R} \) is a continuously differentiable function and \( X \colon [0, T] \to \mathbb{R}^n \), \( t \mapsto X_t \), is a continuously differentiable path. Then the integral
    \[
        \int_0^T f(X_t) \, dX_t
    \]
    is well defined and may be approximated as follows.

    By Taylor's theorem, we have the approximation
    \[
        f(X_t) \approx f(X_s) + \nabla f(X_s) \cdot (X_t - X_s),
    \]
    which becomes increasingly accurate as \(|t - s|\) becomes small. Using this, the integral
    \[
        \int_s^t f(X_u) \, dX_u \approx f(X_s)(X_t - X_s) + \int_s^t \nabla f(X_s) \cdot (X_u - X_s) \, dX_u
    \]
    can likewise be approximated to arbitrary precision. In other words, the linearization
    \[
        f(X_t) \approx f(X_s) + \nabla f(X_s) \cdot (X_t - X_s)
    \]
    provides a good approximation for evaluating the integral of \(f\) against the path \(X_t\).

    Now define the bilinear map
    \[
        (u \otimes v) \mapsto f'(X_s)(u \otimes v) \coloneqq (\nabla f(X_s) \cdot u)v
    \]
    for \( u, v \in \mathbb{R}^n \). With this notation, we may rewrite the second integral as
    \begin{equation*}
    \begin{split}
        \int_s^t \nabla f(X_s) \cdot (X_u - X_s) \, dX_u 
        &= \int_s^t f'(X_s)\big((X_u - X_s) \otimes dX_u\big) \\
        &= f'(X_s) \int_s^t (X_u - X_s) \otimes dX_u \\
        &= f'(X_s) X_{s,t}^{(2)},
    \end{split}
    \end{equation*}
    so that
    \[
        \int_s^t f(X_u) \, dX_u \approx f(X_s)(X_t - X_s) + f'(X_s) X_{s,t}^{(2)}.
    \]

    Consequently, for a sufficiently fine partition \( \mathcal{D} = \{ 0 = t_0 < \cdots < t_n = T \} \) of \([0, T]\), we obtain the approximation
    \[
        \int_0^T f(X_u) \, dX_u \approx \sum_{t_i, t_{i+1} \in \mathcal{D}} \left[ f(X_{t_i})(X_{t_{i+1}} - X_{t_i}) + f'(X_{t_i}) X^{(2)}_{t_i, t_{i+1}} \right],
    \]
    which becomes arbitrarily accurate as the mesh of the partition tends to zero.

    Moreover, it can be verified -- through a somewhat tedious but straightforward computation -- that
    \[
        X_{s,t}^{(2)} = X_{s,u}^{(2)} + X_{u,t}^{(2)} + X_{s,u} \otimes X_{u,t}
    \]
    for all \( s \leq u \leq t \). This identity is known as \emph{Chen's relation}, and it plays a foundational role in the definition of rough integration.

    Additionally, we define the process \( Y'_t \coloneqq f'(X_t) \) as the \emph{Gubinelli derivative} of the path \( Y_t \coloneqq f(X_t) \).

    It is precisely the structure provided by \emph{Chen's relation} and \emph{Gubinelli derivatives} that enables the extension of integration theory to a broader class of irregular paths.
\end{example}

We now formalize the framework introduced in the preceding example. Let \(\pvarSpace{p}{J}{\mathbb{R}^n}\) denote the space of continuous paths \(X \colon J \to \mathbb{R}^n\) with finite \(p\)-variation. Define \(\pvarSpace{0,p}{J}{\mathbb{R}^n}\) as the closure of this space with respect to the \(p\)-variation seminorm \(\| \cdot \|_{p, J}\).

For \(p \in [2, 3)\), we define \(\mathscr{C}^p(J, \mathbb{R}^n)\) to be the space of all \(\frac{1}{p}\)-Hölder continuous rough paths \(\bm{\zeta} = (\zeta, \zeta^{(2)})\) satisfying the following properties:
\begin{equation}
    \zeta \colon J \to \mathbb{R}^n,
\end{equation}
\begin{equation}
\label{enhancement}
    \zeta^{(2)} \colon \Delta_{[s, t]} \to \mathbb{R}^n \otimes \mathbb{R}^n, \quad (s, t) \mapsto \zeta^{(2)}_{s, t},
\end{equation}
\begin{equation}
\label{ChensRelation}
    \zeta^{(2)}_{s, t} = \zeta^{(2)}_{s, r} + \zeta^{(2)}_{r, t} + \zeta_{s, r} \otimes \zeta_{r, t}, \quad \text{for all } 0 \le s \le r \le t \le T,
\end{equation}
\begin{equation}
\label{roughHölderNorm}
\begin{split}
    |||\bm{\zeta}|||_{\frac{1}{p}\text{-H\"ol}} &\coloneqq \|\zeta\|_{\frac{1}{p}\text{-H\"ol}} + \|\zeta^{(2)}\|_{\frac{2}{p}\text{-H\"ol}} \\
    &< +\infty,
\end{split}
\end{equation}
where the Hölder seminorms are defined by
\begin{equation}
\label{HölderNorm}
    \|\zeta\|_{\frac{1}{p}\text{-H\"ol}} \coloneqq \sup_{s \neq t \in J} \frac{|\zeta_{s, t}|}{|t - s|^{1/p}},
\end{equation}
\begin{equation}
\label{enhancedHölderNorm}
    \|\zeta^{(2)}\|_{\frac{2}{p}\text{-H\"ol}} \coloneqq \sup_{s \neq t \in J} \frac{|\zeta^{(2)}_{s, t}|}{|t - s|^{2/p}}.
\end{equation}

The tensor product in \eqref{enhancement} and \eqref{ChensRelation} refers to the \emph{Cartesian tensor product}. Equation~\eqref{ChensRelation} is known as \emph{Chen's relation}, and it must hold for all \(r \in [s, t]\). The map \(\zeta^{(2)}\) is referred to as the \emph{iterated integral} or \emph{lift} of \(\zeta\), and the pair \(\bm{\zeta} = (\zeta, \zeta^{(2)})\) is called the \emph{lifted rough path} or simply the \emph{lift} of \(\zeta\).

We now define the \(p\)-variation of a rough path \(\bm{\zeta} = (\zeta, \zeta^{(2)})\), which quantifies the joint regularity of both components.

\begin{definition}
\label{def:def_rough_var}
Let \(\bm{\zeta} = (\zeta, \zeta^{(2)}) \in \mathscr{C}^p(J, \mathbb{R}^n)\). The \(p\)-variation of the lift is defined as follows:
\begin{equation}
\label{eq:enhanced_variation}
    \|\zeta^{(2)}\|_{\frac{p}{2}, J} \coloneqq \left[ \sup_{\mathcal{D}} \sum |\zeta^{(2)}_{t_i, t_{i+1}}|^{\frac{p}{2}} \right]^{\frac{2}{p}},
\end{equation}
\begin{equation}
\label{eq:rough_variation}
    |||\bm{\zeta}|||_{p, J} \coloneqq \|\zeta\|_{p, J} + \|\zeta^{(2)}\|_{\frac{p}{2}, J},
\end{equation}
where the supremum in \eqref{eq:enhanced_variation} is taken over all partitions \(\mathcal{D}\) of \(J = [0, T]\).
\end{definition}

An immediate consequence of \textbf{Definition~\ref{def:def_rough_var}} is that every element of \(\mathscr{C}^p(J, \mathbb{R}^n)\) has finite \(p\)-variation for \(p \in [2, 3)\).

\begin{proposition}
\label{prop:rough_finite_pvar}
If \(\bm{\zeta} = (\zeta, \zeta^{(2)}) \in \mathscr{C}^p(J, \mathbb{R}^n)\), then \(|||\bm{\zeta}|||_{p, J} < +\infty\).
\end{proposition}

\begin{proof}
\normalsize
Observe that
\begin{equation*}
\begin{split}
    \sum_{\mathcal{D}} |\zeta_{t_i, t_{i+1}}|^p 
    &= \sum_{\mathcal{D}} \frac{|\zeta_{t_i, t_{i+1}}|^p}{t_{i+1} - t_i} (t_{i+1} - t_i) \\
    &\le \sum_{\mathcal{D}} \|\zeta\|^p_{\frac{1}{p}\text{-H\"ol}} (t_{i+1} - t_i) \\
    &= \|\zeta\|^p_{\frac{1}{p}\text{-H\"ol}} \sum_{\mathcal{D}} (t_{i+1} - t_i) \\
    &= \|\zeta\|^p_{\frac{1}{p}\text{-H\"ol}} \cdot T \\
    &< +\infty.
\end{split}
\end{equation*}
Taking the supremum over all partitions \(\mathcal{D}\) shows that
\[
    \|\zeta\|^p_{p, J} \le \|\zeta\|^p_{\frac{1}{p}\text{-H\"ol}} \cdot T,
\]
and hence
\[
    \|\zeta\|_{p, J} \le \|\zeta\|_{\frac{1}{p}\text{-H\"ol}} \cdot T^{1/p}
\]
is finite. An analogous argument applied to \(\zeta^{(2)}\) shows that
\[
    \|\zeta^{(2)}\|_{\frac{p}{2}, J} < +\infty.
\]
Thus, combining the above with \eqref{eq:rough_variation}, we conclude that \(\bm{\zeta}\) has finite \(p\)-variation:
\[
    |||\bm{\zeta}|||_{p, J} < +\infty.
\]
\end{proof}

Building on \textbf{Proposition~\ref{prop:rough_finite_pvar}}, we can equip the space \(\mathscr{C}^p(J, \mathbb{R}^n)\) with natural metrics that measure distances between rough paths in terms of their regularity. This leads to the following definition.

\begin{definition}[Rough Path Metrics]
\label{roughPathMetrics}
Let \(\bm{\zeta} = (\zeta, \zeta^{(2)})\) and \(\bm{\eta} = (\eta, \eta^{(2)})\) be two rough paths in \(\mathscr{C}^p(J, \mathbb{R}^n)\). We define the \(\frac{1}{p}\)-Hölder and \(p\)-variation metrics on the interval \(J\) by
\begin{equation}
\label{roughHolderMetric}
    \roughMetricHolder{\zeta}{\eta}{p}{J} \coloneqq \|\zeta - \eta\|_{\frac{1}{p}\text{-H\"ol}} + \|\zeta^{(2)} - \eta^{(2)}\|_{\frac{2}{p}\text{-H\"ol}},
\end{equation}
\begin{equation}
\label{roughVariationMetric}
    \roughMetricPvar{\zeta}{\eta}{p}{J} \coloneqq \|\zeta - \eta\|_{p, J} + \|\zeta^{(2)} - \eta^{(2)}\|_{\frac{p}{2}, J}.
\end{equation}
It is straightforward to verify that both expressions define metrics on \(\mathscr{C}^p(J, \mathbb{R}^n)\).
\end{definition}

\begin{remark}[Canonical Lift]
\label{canonicalLift}
As discussed in \textbf{Example~\ref{example:mechanics_of_rough_integration}}, any smooth path \(\zeta \colon J \to \mathbb{R}^n\) admits a canonical lift. Specifically, the second-level component is defined by
\[
    \zeta^{(2)}_{s,t} \coloneqq \int_s^t \zeta_{s,r} \otimes d\zeta_r.
\]
This canonical lift plays an important role in connecting classical integration theory to the rough path framework.
\end{remark}

Before proceeding, we pause to clarify a point of terminology. The word ``control” in the next definition should not be confused with its usage in optimal control theory. In the context of rough path theory, its meaning will be clear from the definition and usage that follows.

\begin{definition}[Controlled Rough Paths {\cite{gubinelliControlledRoughPath}}]
\label{def:controlled_rough_paths}
Let \(\bm{\zeta} \in \mathscr{C}^p(J, \mathbb{R}^d)\) be a rough path. The space of \emph{controlled rough paths} with respect to \(\bm{\zeta}\), denoted \(\mathscr{D}_{\zeta}^{p}(J, \mathbb{R}^m)\), consists of all pairs
\[
    (X, X') \in \pvarSpace{p}{J}{\mathbb{R}^m} \times \pvarSpace{p}{J}{\mathcal{L}(\mathbb{R}^d, \mathbb{R}^m)}
\]
such that the remainder term
\[
    R^X_{s,t} \coloneqq X_{s,t} - X'_s \zeta_{s,t}
\]
has finite \(\frac{p}{2}\)-variation on \(J\). The path \(X'\) is called the \emph{Gubinelli derivative} of \(X\) with respect to \(\bm{\zeta}\).
\end{definition}

\begin{remark}[Norm on \(\mathscr{D}^p\)]
The space \(\mathscr{D}_{\zeta}^{p}(J, \mathbb{R}^m)\), equipped with the norm
\[
    \|(X, X')\| \coloneqq |X_0| + |X'_0| + \|X'\|_{p, J} + \|R^X\|_{\frac{p}{2}, J},
\]
is a Banach space.
\end{remark}

We are now in a position to state the existence and uniqueness theorem for rough integrals. A detailed proof can be found in \cite{gubinelliControlledRoughPath}, as it lies beyond the scope of this paper.

\begin{theorem}[Rough Integration]
\label{theorem:rough_integration}
Let \(J = [0, T]\), and suppose \(p \in [2, 3)\), \(\bm{\zeta} \in \mathscr{C}^p(J, \mathbb{R}^d)\), and \((X, X') \in \mathscr{D}_{\zeta}^p(J, \mathcal{L}(\mathbb{R}^d, \mathbb{R}^m))\). Then the rough integral of \(X\) against \(\bm{\zeta}\) is defined by
\begin{equation}
\label{eq:rough_integral_equation}
    \int_0^T X_r \, d\bm{\zeta}_r \coloneqq \lim_{|\mathcal{D}| \to 0} \sum_{\mathcal{D}} \left[ X_{t_i} \zeta_{t_i, t_{i+1}} + X'_{t_i} \zeta^{(2)}_{t_i, t_{i+1}} \right],
\end{equation}
where the limit is taken over all partitions \(\mathcal{D} = \{t_0, \dots, t_n\}\) of \([0, T]\) as the mesh size tends to zero. This limit exists, is independent of the choice of partition, and defines a unique value. The resulting object is called the \emph{rough integral} of \(X\) with respect to \(\bm{\zeta}\).

Moreover, for all \(0 \le s \le t \le T\), the rough integral satisfies the following remainder estimate:
\begin{equation}
\label{eq:integral_remainder}
    \left| \int_s^t X_r \, d\bm{\zeta}_r - X_s \zeta_{s, t} - X'_s \zeta^{(2)}_{s, t} \right|
    \le C \left( \|R^X\|_{\frac{p}{2}, [s, t]} \|\zeta\|_{p, [s, t]} + \|X'\|_{p, [s, t]} \|\zeta^{(2)}\|_{\frac{p}{2}, [s, t]} \right),
\end{equation}
where \(C > 0\) is a constant depending only on \(p\).
\end{theorem}

\subsection{Integration Theory for Optimal Control}
\label{sec:integration_theory_for_optimal_control}
In the setting of \emph{stochastic} optimal control, one frequently encounters stochastic differential equations (SDEs) of the form \[ dX_s = b(X_s, \gamma_s) \, ds + \sigma(X_s, \gamma_s) \, dB_s, \] where \( \gamma_t \) denotes a control process and \( B_t \) is a standard Brownian motion. Here, \( b \) and \( \sigma \) describe the drift and diffusion coefficients, respectively, which depend on both the state \( X_s \) and the control \( \gamma_s \). In this classical setting, optimal control strategies can be derived under the probabilistic framework of stochastic processes, with well-established methods for handling the interplay between the stochastic behavior of \( B_t \) and the influence of the control.

However, attempting to naively extend these principles of optimal control to the \emph{rough differential equation} (RDE) framework presents substantial challenges. Specifically, when pathwise optimal control techniques are applied directly to rough paths, degeneracy issues often arise. This phenomenon, thoroughly investigated by Diehl, Friz, and Gassiat \cite{DiehlFrizGassiat} as well as by Allan and Cohen \cite{allan2019pathwise}, reveals that traditional methods fail to preserve the desired properties in the rough setting, leading to breakdowns in the stability and effectiveness of the control processes. These degeneracies highlight a fundamental gap between classical SDE methods and rough path theory when applied to control problems, necessitating alternative approaches for meaningful analysis and application in the rough path setting.

In order to overcome these limitations, we shift our attention to rough differential equations (RDEs) of the form \[ dX_s = b(X_s, \gamma_s) \, ds + \lambda(X_s, \gamma_s) \, d\bm{\zeta}_s, \] where \( \bm{\zeta}_s \) is a rough path serving as a generalized driving signal. The dependence of the diffusion coefficient \( \lambda \) on both \( X_s \) and the control process \( \gamma_s \) introduces further complexities, moving us out of the classical framework. This dependence presents a novel challenge, as the rough path framework demands a reformulation of control techniques that accounts for the irregularity of \( \bm{\zeta} \) while managing the interactions between the pathwise nature of \( \bm{\zeta} \) and the control.

In the next section, we will address these issues in detail, exploring methods to resolve the degeneracies associated with pathwise optimal control in the rough path context. By developing a more refined approach, we aim to establish a consistent framework for RDE-based optimal control that accommodates both the roughness of the driving signal and the intricate dependencies introduced by the control variable. This framework will provide a foundation for further research and application in rough path control, advancing the theory beyond the limitations encountered in earlier approaches.

\begin{definition}[Rough Differential Equations]
\label{RDEdefinition}
Let \(p \in [2, 3)\), \(\bm{\zeta} \in \mathscr{C}^p(J, \mathbb{R}^d)\), and \[ \gamma \in \pvarSpace{\frac{p}{2}}{J}{\mathbb{R}^k}. \] We consider rough differential equations (RDEs) of the form
\begin{equation}
\label{RDEequation}
    dX_s = b(X_s, \gamma_s) \, ds + \lambda(X_s, \gamma_s) \, d\bm{\zeta}_s,
\end{equation}
with initial condition \(X_0 = x_0\), where the integral with respect to \(\bm{\zeta}\) is defined in the sense of \textbf{Theorem~\ref{theorem:rough_integration}}.
\end{definition}

\begin{remark}[Gubinelli Derivative]
Suppose \((X, X') \in \mathscr{D}^p\) and \(\lambda \in C_b^2\). Then the composition \(\lambda(X, \gamma)\) has Gubinelli derivative
\[
    \lambda(X, \gamma)' = \partial_x \lambda(X, \gamma) \, X',
\]
where \(\partial_x \lambda\) denotes the Fréchet derivative of \(\lambda\) with respect to its first argument. Specifically, \(\partial_x \lambda(x, \gamma) = A\) for a linear transformation \(A\) such that
\[
    \lim_{|h| \to 0} \ \frac{|\lambda(x + h, \gamma) - \lambda(x, \gamma) - Ah|}{|h|} = 0.
\]
See \cite{sternberg} for further details.
\end{remark}

To support further analysis of pathwise regularity, we now state some classical results on Hölder continuity, variation spaces, and related properties. These will be used in later arguments.

\begin{proposition}
\label{prop:holder_continous_paths_have_pvar}
Let \(\alpha \in (0, 1)\). Then every \(\alpha\)-Hölder continuous path has finite \(\frac{1}{\alpha}\)-variation.
\end{proposition}

\begin{proof}
\normalsize
By \(\alpha\)-Hölder continuity, we mean that \(X \colon J \to \mathbb{R}^n\) satisfies
\[
    \sup_{s \neq t \in J} \frac{|X_{s, t}|}{|t - s|^{\alpha}} < +\infty.
\]
Applying the same argument as in \textbf{Proposition~\ref{prop:rough_finite_pvar}} yields the result.
\end{proof}

The following inclusion result is standard and stated without proof.

\begin{lemma}
\label{lemma:variation_inclusion}
If \(1 \le p \le q < +\infty\), then
\[
    \pvarSpaceSimple{1} \subseteq \pvarSpaceSimple{p} \subseteq \pvarSpaceSimple{q}.
\]
\end{lemma}

We now show that the map \(t \mapsto \|X\|_{1, [0, t]}\) itself has finite \(p\)-variation for any \(p \in [1, \infty)\).

\begin{lemma}
\label{lemma:1var_has_finite_pvar}
Let \(X \in \pvarSpace{1}{J}{\mathbb{R}^n}\). Then the function \(t \mapsto \|X\|_{1, [0, t]}\) has finite \(p\)-variation for all \(1 \le p < +\infty\).
\end{lemma}

\begin{proof}
\normalsize
Define the function
\[
    f(t) \coloneqq \|X\|_{1, [0, t]}.
\]
It is straightforward to verify that \(f\) is monotonically increasing on \(J\). Hence, for any partition \(\mathcal{D} = \{0 = t_0 < \dots < t_n = T\}\), we have
\begin{equation*}
\begin{split}
    \sum_{\mathcal{D}} |f(t_{i+1}) - f(t_i)| 
    &= \sum_{\mathcal{D}} f(t_{i+1}) - f(t_i) \\
    &= f(T) - f(0) \\
    &< +\infty.
\end{split}
\end{equation*}
This shows that \(f\) has finite 1-variation. The result now follows from \textbf{Lemma~\ref{lemma:variation_inclusion}}, which implies that any function of finite 1-variation also has finite \(p\)-variation for \(p \ge 1\).
\end{proof}

We now collect several auxiliary inequalities that will be used in subsequent analysis. The first is a simple consequence of Jensen's inequality.

\begin{lemma}
\label{lemma:inequality1}
Let \(p \ge 1\), and let \(x_1, \dots, x_n\) be non-negative real numbers. Then
\[
    (x_1 + \dots + x_n)^p \le n^p(x_1^p + \dots + x_n^p).
\]
\end{lemma}

\begin{proof}
\normalsize
Define a probability measure on \(\Omega = \{x_1, \dots, x_n\}\) by setting
\[
    \mathbb{P}\{x_i\} = \frac{1}{n}.
\]
Since the function \(t \mapsto t^p\) is convex on \([0, \infty)\), Jensen’s inequality implies
\begin{equation*}
\begin{split}
    \left( \frac{x_1 + \dots + x_n}{n} \right)^p 
    &= \frac{1}{n^p} (x_1 + \dots + x_n)^p \\
    &\le \frac{1}{n} (x_1^p + \dots + x_n^p).
\end{split}
\end{equation*}
Multiplying both sides by \(n^p > 1\) yields
\[
    (x_1 + \dots + x_n)^p \le n^{p - 1}(x_1^p + \dots + x_n^p) \le n^p(x_1^p + \dots + x_n^p).
\]
\end{proof}

The next two lemmas provide tools for bounding variation norms.

\begin{lemma}
\label{lemma:inequality2}
Let \(X \in \pvarSpaceSimple{p}\), and let \(\mathcal{D} = \{t_0 < \dots < t_n\}\) be a partition of \(J\). Then
\[
    \|X\|_{p, J} \le n \left( \sum_{i = 0}^{n-1} \|X\|^p_{p, [t_i, t_{i+1}]} \right)^{\frac{1}{p}}.
\]
\end{lemma}

\begin{proof}
\normalsize
Fix a partition \(\mathcal{D} = \{t_0 < \dots < t_n\}\) of \(J\), and for each \(i = 0, \dots, n-1\), let \[\mathcal{D}_i = \{s^i_0 < \dots < s^i_{k_i}\}\] be a refinement of \([t_i, t_{i+1}]\). Then by subadditivity and \textbf{Lemma~\ref{lemma:inequality1}}, we have
\begin{equation*}
\begin{split}
    \sum_{\mathcal{D}} |X_{t_{i+1}} - X_{t_i}|^p 
    &\le \sum_{i=0}^{n-1} \left( \sum_{j=0}^{k_i - 1} |X_{s^i_{j+1}} - X_{s^i_j}| \right)^p \\
    &\le n^p \sum_{i=0}^{n-1} \sum_{j=0}^{k_i - 1} |X_{s^i_{j+1}} - X_{s^i_j}|^p \\
    &\le n^p \sum_{i=0}^{n-1} \|X\|_{p, [t_i, t_{i+1}]}^p.
\end{split}
\end{equation*}
Taking the \(p\)th root on both sides gives the desired inequality.
\end{proof}

\begin{lemma}
\label{lemma:inequality3}
Let \(d_1, \dots, d_n \in \mathbb{N}_+\), and let \(x_i \in \mathbb{R}^{d_i}\). Then
\[
    |(x_1, \dots, x_n)| \le |x_1| + \dots + |x_n|,
\]
where \(|\cdot|\) denotes the Euclidean norm on each respective space.
\end{lemma}

\begin{proof}
\normalsize
The result follows by repeated application of the triangle inequality. Specifically, observe that
\[
    (x_1, \dots, x_n) = (x_1, 0, \dots, 0) + (0, x_2, 0, \dots, 0) + \dots + (0, \dots, 0, x_n),
\]
and apply the triangle inequality in \(\mathbb{R}^{d_1 + \dots + d_n}\).
\end{proof}

To simplify notation, we write \(|X_{s,t}| \lesssim |t - s|\) to indicate that \(|X_{s,t}| \le C|t - s|\) for some constant \(C > 0\). The following proposition establishes several regularity results for rough differential equations. The expanded and detailed proof is provided in the appendix.
\begin{proposition}[Regularity]
\label{prop:regularity_RDE}
    Suppose that $p \in [2, 3)$, $\bm{\zeta} \in \mathscr{C}^p$, $M \in \bb{R}_+$, $\roughholder{\zeta}{1}{p}{J} \le M$, $b \in Lip_b$, $\psi, \lambda \in C^2_b$, $\gamma \in \pvarSpaceSimple{\frac{p}{2}}$ and $X$ satisfies the RDE (\ref{RDEequation}) with Gubinelli derivative $X^{\prime} = \lambda(X, \gamma)$. Then we have the following regularity results
    \begin{enumerate}
        \item \label{est1} $\pvar{\psi(X, \gamma)^{\prime}}{p}{J} \lesssim \pvar{X}{p}{J} + \pvar{\gamma}{\frac{p}{2}}{J}$
        \item \label{est2} $\pvar{R^\psi}{\frac{p}{2}}{J} \lesssim \pvar{X}{p}{J}^2 + \pvar{R^X}{\frac{p}{2}}{J} + \pvar{\gamma}{\frac{p}{2}}{J}$
        \item \label{est3} $\pvar{X}{p}{J} \lesssim 1 + \pvar{\gamma}{\frac{p}{2}}{J}^{1 + p}$
        \item \label{est4} $\pvar{R^X}{\frac{p}{2}}{J} \lesssim 1 + \pvar{\gamma}{\frac{p}{2}}{J}^{2 + p}.$
    \end{enumerate}
\end{proposition}

We conclude this section with two key theorems, whose proofs can be found in \cite{allan2019pathwise}. Although the proofs are lengthy, they are conceptually straightforward, relying on standard arguments and the results established above.

\begin{theorem}[Existence \& Uniqueness]
\label{theorem:RDETheorem1}
Suppose that \(b \in \mathrm{Lip}_b\), \(\lambda \in C^3_b\), and \(\bm{\zeta} \in \mathscr{C}^p\). If \(\gamma \in \pvarSpaceSimple{\frac{p}{2}}\) and \(x\) is fixed, then there exists a unique solution \((X, X') \in \mathscr{D}^p_{\zeta}\) to the RDE
\[
    dX_t = b(X_t, \gamma_t)\,dt + \lambda(X_t, \gamma_t)\,d\bm{\zeta}_t,
\]
with initial condition \(X_0 = x\).
\end{theorem}

\begin{theorem}
\label{theorem:RDETheorem2}
Suppose \(b \in \mathrm{Lip}_b\), \(\lambda \in C^3_b\), and that two rough paths \(\bm{\zeta}, \bm{\eta} \in \mathscr{C}^p\) satisfy \(\roughholder{\zeta}{1}{p}{} , \roughholder{\eta}{1}{p}{} \le M\) for some \(M > 0\). If \(\gamma, \vartheta \in \pvarSpaceSimple{\frac{p}{2}}\) and the following inclusions
\((X, X') = (X, \lambda(X, \gamma)) \in \mathscr{D}^p_{\zeta}\), 
\((Y, Y') = (Y, \lambda(Y, \vartheta)) \in \mathscr{D}^p_{\eta}\) hold,
then
\[
    \pvar{X' - Y'}{p}{J} \lesssim |x - y| + \|\gamma - \vartheta\|_{\infty, J} + \pvar{\gamma - \vartheta}{p}{J} + \roughMetricPvar{\zeta}{\eta}{p}{J}.
\]
Moreover, if \(\psi \in C^3_b\), then
\begin{equation}
\label{eq:RDEDistance} 
\begin{split}
\left\|\int_0^\cdot \psi(X_s, \gamma_s)\, d\bm{\zeta}_s - \int_0^\cdot \psi(Y_s, \vartheta_s)\, d\bm{\eta}_s \right\|_{p, J} 
& \lesssim |x - y| + \|\gamma - \vartheta\|_{\infty, J} + \pvar{\gamma - \vartheta}{p}{J} \\ 
& + \roughMetricPvar{\zeta}{\eta}{p}{J}. \\
\end{split}
\end{equation}
\end{theorem}

\subsection{Signatures}
\label{sec:signatures}

We now turn to the notion of \emph{path signatures}, which generalizes the concept of a path's \emph{lift}. Recall that the second-level lift of a continuously differentiable path \(X_t\) is defined by
\[
    \mathbb{X}_{s,t}^{(2)} \coloneqq \int_s^t (X_u - X_s) \otimes dX_u.
\]
If higher-order derivatives exist for both \(f\) and \(X_t\), we can extend the idea from \textbf{Example~\ref{example:mechanics_of_rough_integration}} to incorporate the additional structure revealed by higher-order iterated integrals.

Suppose that \(f\) and \(X_t\) in \textbf{Example~\ref{example:mechanics_of_rough_integration}} are at least \(C^k\). Then, by Taylor’s theorem, we can approximate \(f(X_t)\) as
\[
    f(X_t) \approx f(X_s) + \sum_{l=1}^k \frac{1}{l!} d^l f(X_s)\, (X_t - X_s)^{\otimes l},
\]
where \(d^l f(X_s)\) denotes the \(l\)th differential of \(f\) evaluated at \(X_s\). Consequently, the integral of \(f(X_t)\) against \(X_t\) may be approximated by
\begin{equation*}
\begin{split}
    \int_0^T f(X_t)\, dX_t & \approx \sum_{t_i, t_{i+1} \in \mathcal{D}} \left[ \sum_{l=0}^k \int_{t_i}^{t_{i+1}} \left(\frac{1}{l!} d^l f(X_{t_i})\, (X_t - X_{t_i})^{\otimes l} \right) dX_t \right] \\
    & = \sum_{t_i, t_{i+1} \in \mathcal{D}} \left[ \sum_{l=0}^k f^{(l)}(X_{t_i}) \int_{t_i < u_1 < \cdots < u_l < t_{i+1}} dX_{u_1} \otimes \cdots \otimes dX_{u_l} \right] \\
    & = \sum_{t_i, t_{i+1} \in \mathcal{D}} \left[ \sum_{l=0}^k f^{(l)}(X_{t_i})\, \mathbb{X}_{t_i, t_{i+1}}^{(l)} \right].
\end{split}
\end{equation*}
Here, \(\mathbb{X}_{s,t}^{(l)} \in (\mathbb{R}^n)^{\otimes l}\) represents the \(l\)th-level iterated integral of \(X_t\), encoding higher-order information that refines the approximation of the integral. The collection
\[
    \mathbb{X}_{s,t}^{\le k} \coloneqq \left(1, X_{s,t}, \mathbb{X}_{s,t}^{(2)}, \dots, \mathbb{X}_{s,t}^{(k)} \right)
\]
is called the \emph{\(k\)-truncated signature} of the path \(X_t\), while the full sequence
\[
    \mathbb{X}_{s,t}^{< \infty} \coloneqq \left(1, X_{s,t}, \mathbb{X}_{s,t}^{(2)}, \dots \right)
\]
is referred to as the \emph{signature} of the path.

Throughout this section, let \(V\) denote a finite-dimensional vector space over \(\mathbb{R}\), and fix \(J = [0, T]\). The theory developed here will later be applied in the context of optimal stopping with rough paths. For a comprehensive introduction to signatures, we refer the reader to \cite{lyons}.

Before introducing the formal definitions, we state a basic but useful estimate concerning the size of iterated integrals.

\begin{proposition}
    If \(X \in \pvarSpace{1}{J}{V}\), then for all \(k \in \mathbb{N}_+\),
    \begin{equation*}
        \left| \int_{0 < u_1 < \cdots < u_k < T} dX_{u_1} \otimes \cdots \otimes dX_{u_k} \right| \le \frac{\pvar{X}{1}{J}^k}{k!}.
    \end{equation*}
\end{proposition}

\begin{proof}
\normalsize
    The result follows by observing that \(|X_{0,u_i}| \le \pvar{X}{1}{J}\), and expanding the integrand via
    \begin{equation*}
        \left| \int_{0 < u_1 < \cdots < u_k < T} dX_{u_1} \otimes \cdots \otimes dX_{u_k} \right| = \left| \int_{0 < u_1 < \cdots < u_k < T} \dot{X}_{u_1} \otimes \cdots \otimes \dot{X}_{u_k} \, du_1 \cdots du_k \right|.
    \end{equation*}
\end{proof}

\subsubsection{Tensor Algebras}
\label{subsec:tensor_algebras}
We adopt the convention \(V^{\otimes 0} = \mathbb{R}\), and now define the extended tensor algebra.

\begin{definition}[Extended Tensor Algebra]
\label{def:extendedTensorAlgebra}
    The \emph{extended tensor algebra} over a finite-dimensional vector space \(V\) is defined as
    \begin{equation}
    \label{tensorSeries}
        T((V)) \coloneqq \left\{ \bm{a} = (a_0, a_1, a_2, \dots) : a_n \in V^{\otimes n},\ n \in \mathbb{N} \right\},
    \end{equation}
    with addition given by
    \begin{equation}
    \label{tensorSeriesAddition}
        \bm{a} + \bm{b} \coloneqq (a_0 + b_0, a_1 + b_1, a_2 + b_2, \dots),
    \end{equation}
    and multiplication (tensor product) defined by
    \begin{equation}
    \label{tensorSeriesProduct}
        \bm{a} \otimes \bm{b} \coloneqq (c_0, c_1, c_2, \dots),
    \end{equation}
    where each component \(c_n\) is given by
    \begin{equation}
    \label{cn}
        c_n \coloneqq \sum_{k=0}^n a_k \otimes b_{n-k}.
    \end{equation}

    For simplicity, we write \(\bm{a} \otimes \bm{b}\) as \(\bm{a}\bm{b}\), and \(\bm{a}^{\otimes n}\) as \(\bm{a}^n\). Scalar multiplication by \(\lambda \in \mathbb{R}\) is defined component-wise:
    \[
        \lambda \bm{a} \coloneqq (\lambda a_0, \lambda a_1, \lambda a_2, \dots).
    \]
    The unit element in \(T((V))\) is defined by \(\bm{1} \coloneqq (1, 0, 0, \dots)\).
\end{definition}

Suppose now that \(\bm{a} \in T((V))\) and \(a_0 \neq 0\). Then the (multiplicative) inverse of \(\bm{a}\) is given by the geometric series
\begin{equation}
\label{tensorSeriesInverse}
    \bm{a}^{-1} = \frac{1}{a_0} \sum_{n=0}^\infty \left( \bm{1} - \frac{\bm{a}}{a_0} \right)^n,
\end{equation}
where \(a_0 \in \mathbb{R} \setminus \{0\}\).

In the context of \emph{optimal stopping with signatures}, we will work with the \emph{truncated tensor algebra}:
\begin{equation}
\label{truncatedTensorAlgebra}
    T^N(V) \coloneqq \bigoplus_{n=0}^N V^{\otimes n},
\end{equation}
as practical computations are performed in this finite-dimensional space.

We also define the canonical projection maps:
\begin{equation}
\label{pin}
    \pi_n \colon T((V)) \to V^{\otimes n},\quad \bm{a} \mapsto \pi_n(\bm{a}) = a_n,
\end{equation}
and the truncated projection
\begin{equation}
\label{piN}
    \pi_{\le N} \colon T((V)) \to T^N(V),\quad \bm{a} \mapsto \pi_{\le N}(\bm{a}) = (a_0, \dots, a_N).
\end{equation}

\subsubsection{Norms}
On \(V\), we use the \(\ell^\infty\) norm. That is, if \(e_1, \dots, e_n\) denote the standard basis of \(V\), and \(v \in V\) is given by
\[
v = \lambda_1 e_1 + \cdots + \lambda_n e_n,
\]
then \(|v| \coloneqq \max_{1 \le i \le n} |\lambda_i|\). On the dual space \(V^*\), we use the \(\ell^1\) norm, defined by \[|v^*| = |\lambda_1| + \cdots + |\lambda_n|.\]

We also adopt the corresponding norms on the tensor powers of \(V\) and \(V^*\). Specifically, for \(\bm{a} \in T((V))\), we define
\[
|\bm{a}| \coloneqq \sup_{i \in \mathbb{N}} |a_i|,
\]
and for \(\bm{b} \in \bigoplus_{n=0}^\infty T(V^*)\), we define
\[
|\bm{b}| \coloneqq \sum_{i=0}^{\infty} |b_i|.
\]

\subsubsection{Shuffles}
Let \(e_{i_1}^*, \dots, e_{i_n}^*\) be standard basis vectors in \(T(V^*)\). We denote the tensor product \(e_{i_1}^* \otimes \cdots \otimes e_{i_n}^*\) by the word \(i_1 \cdots i_n\).

Let \(\mathcal{W}(\mathcal{A}_d)\) denote the linear span of all finite words formed from the alphabet \(\mathcal{A}_d \coloneqq \{1, \dots, d\}\). The empty word is denoted by \(\emptyset\), and scalar multiplication by elements of \(\mathbb{R}\) is defined in the usual way. For words \(l_1, l_2 \in \mathcal{W}(\mathcal{A}_d)\), we define addition as concatenation:
\[
l_1 + l_2 \coloneqq l_1l_2 \in \mathcal{W}(\mathcal{A}_d).
\]

The degree of a word is defined by
\begin{equation}
\label{wordDegree}
\deg(w) \coloneqq
\begin{dcases}
    n, & \text{if } w = i_1 \cdots i_n, \\
    \max_{1 \le j \le n} \deg(w_j), & \text{if } w = \lambda_1 w_1 + \cdots + \lambda_n w_n.
\end{dcases}
\end{equation}

We now define the \emph{shuffle product} on \(\mathcal{W}(\mathcal{A}_d)\) recursively by:
\[
w \shuffle \emptyset \coloneqq \emptyset \shuffle w \coloneqq w,
\]
\[
wi \shuffle vj \coloneqq (w \shuffle vj)i + (wi \shuffle v)j.
\]
This product is then extended bilinearly to all of \(\mathcal{W}(\mathcal{A}_d) \times \mathcal{W}(\mathcal{A}_d)\). The shuffle product \(\shuffle\) is associative, commutative, and distributive over addition.

\begin{example}[Shuffle Product]
A simple example of the shuffle product is
\begin{equation*}
\begin{split}
    12 \shuffle 3 &= (1 \shuffle 3)2 + (12 \shuffle \emptyset)3 \\
    &= (1 \shuffle 3)2 + 123 \\
    &= ((\emptyset \shuffle 3)1 + (1 \shuffle \emptyset)3)2 + 123 \\
    &= 312 + 132 + 123,
\end{split}
\end{equation*}
since \(\emptyset w = w = w \emptyset\).
\end{example}

Given a polynomial \(P(x) = \lambda_0 + \lambda_1 x + \cdots + \lambda_n x^n\) in a commuting variable \(x\), we define the associated \emph{shuffle polynomial} by
\[
P^{\shuffle}(l) \coloneqq \lambda_0 \emptyset + \lambda_1 l + \cdots + \lambda_n l^{\shuffle n}.
\]

Lastly, define the pairing
\begin{equation}
\label{tensorPairing}
\inner{\cdot}{\cdot} \colon T(V^*) \times T((V)) \to \mathbb{R}, \quad (v^*, w) \mapsto \inner{v^*}{w} = v^*(w),
\end{equation}
and set
\begin{equation}
\label{G}
G(V) \coloneqq \left\{ \bm{a} \in T((V)) \setminus \{ \bm{0} \} : \inner{l_1 \shuffle l_2}{\bm{a}} = \inner{l_1}{\bm{a}} \inner{l_2}{\bm{a}},\ \forall\ l_1, l_2 \in T(V^*) \right\}.
\end{equation}
The set \(G(V)\) is called the set of \emph{group-like elements} of \(T((V))\).

\subsubsection{Rough Path Signatures}
\label{subsec:rough_path_signatures}
A \emph{weakly geometric} $p$\emph{-rough path} is a continuous path $\bb{X} \colon [0, T] \to G^{\lfloor p \rfloor} \subset T^{\lfloor p \rfloor}(V)$ with $\bb{X}_0 = 1$ and finite $p$-variation, defined more generally by
\[
|\bb{X}|_{p\text{-var};[0,T]} \coloneqq \max_{k=1,\dots,\lfloor p \rfloor} \sup_{\mathcal{D} \subset [0,T]} \left( \sum_{t_i \in \mathcal{D}}\left|\pi_k(\bb{X}_{t_i,t_{i+1}})\right|^{\frac{p}{k}} \right)^{\frac{k}{p}}.
\]
We denote the space of weakly geometric $p$-rough paths by $\mathcal{W}\Omega^p_T$, equipped with the metric
\[
d_{p\text{-var};[0,T]}(\bb{X},\bb{Y}) \coloneqq |\bb{X} - \bb{Y}|_{p\text{-var};[0,T]}.
\]

Given a continuous path $X \colon [0, T] \to V$ of finite $p$-variation, its \emph{signature} is defined as
\[
\bb{X}^{<\infty} \colon [0, T] \to T((V)),
\]
the sequence of its iterated integrals. That is,
\[
\pi_k(\bb{X}^{<\infty}) \coloneqq \int_{0 < u_1 < \cdots < u_k = t} dX_{u_1} \otimes \cdots \otimes dX_{u_k}
\]
for each $k \in \mathbb{N}$, with $\pi_0(\bb{X}^{<\infty}) \coloneqq 1$.

The space of \emph{geometric} $p$\emph{-rough paths}, denoted $\Omega_T^p$, consists of all paths $\bb{X} \in \mathcal{W}\Omega_T^p$ that can be approximated in $\lfloor p \rfloor$-variation by a sequence of piecewise smooth paths.

Finally, we define $\Hat{\Omega}^p_T$ as the closure (in $p$-variation) of the canonical lifts
\[
\Hat{\bb{X}}^{\le \lfloor p \rfloor} \coloneqq \pi_{\le \lfloor p \rfloor}(\Hat{\bb{X}}^{<\infty}),
\]
where $X$ ranges over piecewise smooth $\mathbb{R}^d$-valued paths and $\Hat{X}_t \coloneqq (t, X_t)$.

\section{Pathwise Optimal Control}
\label{OC}

Before we explore optimal control from a pathwise perspective, it is instructive to consider a striking phenomenon observed by several authors, including Allan \cite{allan2019pathwise} and Diehl, Friz, and Gassiat \cite{DiehlFrizGassiat}: under certain conditions, the value function can diverge to infinity. A particularly illustrative case arises in the context of a trader seeking to maximize profit without sufficient constraints, which can result in unbounded growth of the value function. This example highlights the critical importance of imposing appropriate constraints in optimal control problems -- constraints that are essential for ensuring meaningful and stable solutions in real-world applications. Understanding such degeneracies is a necessary step toward developing robust methods in the pathwise optimal control framework.

\begin{example}[Degeneracy]
\label{example:degeneracy}
Consider an agent attempting to maximize profit by trading a stock. Suppose that instead of relying on discrete-time price predictions of the form
\[
S_{t_{i+1}} = f(\varphi_{t_i})S_{t_i} + \varepsilon_{t_i},
\]
the agent has access to a continuous-time prediction $\zeta_t$, so that the price $S_t$ evolves according to
\[
dS_t = d\zeta_t + \sigma\, dB_t,
\]
where $f$ is a predictive model based on features $\varphi_{t_i}$, $\varepsilon_{t_i}$ is additive noise, $B_t$ is a standard Brownian motion representing continuous-time noise, and $\zeta_t$ is a time-dependent signal capturing the agent's prediction.

Assume the following conditions:
\begin{enumerate}
    \item $S_0 = \zeta_0 = s$;
    \item \label{ass:infinite_1var} $\zeta_t$ has infinite 1-variation over the interval $J = [0, T]$;
    \item The agent may hold up to $Q \in \mathbb{R}_+$ units of stock, either long or short;
    \item $q_t \in [-Q, Q]$ denotes the inventory position at time $t \in [0, T]$;
    \item The agent's wealth evolves as
    \[
    dX_u^{t,x,q} = q_u\, dS_u = q_u\left(d\zeta_u^{t,s} + \sigma\, dB_u\right);
    \]
    \item The value function is defined as
    \[
    v(t,x) = \sup_q \mathbb{E}\left[X_T^{t,x,q}\right].
    \]
\end{enumerate}

To gain insight, approximate the signal $\zeta_t$ by a smooth path $\eta_t$. Then the agent's wealth satisfies
\[
dX_u^{t,x,q} = q_u\left((\eta^{t,s}_u)^\prime\, du + \sigma\, dB_u\right).
\]
In this setting, the Hamilton-Jacobi-Bellman (HJB) equation takes the classical form:
\[
0 = \frac{\partial v}{\partial t} + \sup_q \left[ q\, \eta\, \frac{\partial v}{\partial x} + \frac{1}{2} \sigma^2 q^2\, \frac{\partial^2 v}{\partial x^2} \right].
\]
Assuming a value function of the form $v(t,x) = x + u(t)$ and substituting into the HJB equation yields
\[
0 = \frac{du}{dt} + \sup_q qu.
\]
This implies that the optimal control is
\[
q^*_t \coloneqq Q\, \text{sgn}(\eta^\prime_t),
\]
and the optimal value function becomes
\begin{equation*}
\begin{split}
v^*(t,x) & = x + Q \int_t^T \text{sgn}(\eta_u^\prime)\, \eta_u^\prime\, du \\
        & = x + Q \int_t^T |\eta_u^\prime|\, du \\
        & = x + Q \cdot \text{(path length of $\eta$ over $[t, T]$)}.
\end{split}
\end{equation*}

As $\eta$ becomes a better approximation of $\zeta$, the path length of $\eta$ tends to infinity by \textbf{Assumption \ref{ass:infinite_1var}}, implying the agent can theoretically achieve infinite profit. This clearly contradicts practical reality and reveals a fundamental issue: the absence of appropriate regularization or cost terms in the control problem can result in ill-posed behavior.

\end{example}

\subsection{The Set-Up}
\label{subsec:the_set-up}

To emphasize that we are working with $p \in [2, 3)$, we adopt distinct notation for the relevant rough path spaces. We define the space of geometric rough paths in this regime. Let $\mathscr{C}^{0, p}_g \subset \mathscr{C}^p$ denote the closure, with respect to the $\varrho_{\frac{1}{p}\text{-Höl}}$ topology, of the canonical lifts of smooth paths. This closure is well-defined by the Stone–Weierstrass theorem.

Fix a geometric rough path $\bm{\zeta} \in \mathscr{C}_g^{0, p}(J, \mathbb{R}^d)$. In this section, we study the following optimal control problem:
\begin{equation}
\label{ocProblem1}
    v(t, x) \coloneqq \inf_{\gamma \in \pvarSpaceSimple{p/2}} J(t, x, \gamma),
\end{equation}
where the cost functional is defined by
\begin{equation}
\label{costFunctional}
    J(t, x, \gamma) = \int_t^T f(X_s^{t, x, \gamma}, \gamma_s) \, ds + \int_t^T \psi(X_s^{t, x, \gamma}, \gamma_s) \, d\bm{\zeta}_s + g(X_T^{t, x, \gamma}),
\end{equation}
and $X_s^{t, x, \gamma}$ solves the RDE \eqref{RDEequation} with initial condition $X_t^{t, x, \gamma} = x$. We refer to $v(t, x)$ and $J(t, x, \gamma)$ as the \emph{value function} and \emph{cost functional}, respectively. It is assumed that the functions satisfy
\[
f \colon \mathbb{R}^m \times \mathbb{R}^k \to \mathbb{R}, \quad \psi \colon \mathbb{R}^m \times \mathbb{R}^k \to \mathcal{L}(\mathbb{R}^d, \mathbb{R}), \quad \text{and} \quad g \colon \mathbb{R}^m \to \mathbb{R}.
\]

\begin{lemma}
\label{ineq5}
Suppose the assumptions of \textbf{Theorem~\ref{theorem:RDETheorem1}} hold. Then
\begin{equation}
\label{throwaway1}
    \left| \int_t^T \psi(X_s^{t, x, \gamma}, \gamma_s) \, d\bm{\zeta}_s \right| \lesssim 1 + \pvar{\gamma}{\frac{p}{2}}{[t, T]}^{2(1 + p)}.
\end{equation}
\end{lemma}

\begin{proof}
\normalsize
By \textbf{Theorem~\ref{theorem:rough_integration}} and \textbf{Proposition~\ref{prop:regularity_RDE}}, we have
\begin{equation*}
\begin{split}
    \left| \int_t^T \psi(X_s^{t, x, \gamma}, \gamma_s) \, d\bm{\zeta}_s \right| 
    & \lesssim \pvar{R^\psi}{\frac{p}{2}}{[t, T]} \pvar{\zeta}{p}{[t, T]} + \pvar{\psi^\prime}{p}{[t, T]} \pvar{\lift{\zeta}{}}{\frac{p}{2}}{[t, T]} \\
    & \quad + |\psi(x, \gamma_t) \, \zeta_{t, T}| + |\psi(x, \gamma_t)^\prime \, \lift{\zeta}{t,T}| \\
    & \lesssim \left( \pvar{X^{t, x, \gamma}}{p}{[t, T]}^2 + \pvar{R^X}{\frac{p}{2}}{[t, T]} + \pvar{\gamma}{\frac{p}{2}}{[t, T]} \right) \pvar{\zeta}{p}{[t, T]} \\
    & \quad + \left( \pvar{X^{t,x}}{p}{[t, T]} + \pvar{\gamma}{\frac{p}{2}}{[t, T]} \right) \pvar{\lift{\zeta}{}}{\frac{p}{2}}{[t, T]} \\
    & \quad + |\psi(x, \gamma_t) \, \zeta_{t, T}| + |\psi(x, \gamma_t)^\prime \, \lift{\zeta}{t,T}|.
\end{split}
\end{equation*}
The final two terms are bounded by a constant. Moreover, we have
\begin{equation}
    \pvar{R^X}{\frac{p}{2}}{[t, T]} \lesssim 1 + \pvar{\gamma}{\frac{p}{2}}{[t, T]}^{2 + p}
\end{equation}
and
\begin{equation}
\label{throwaway2}
    \pvar{X}{p}{[t, T]} \lesssim 1 + \pvar{\gamma}{\frac{p}{2}}{[t, T]}^{1 + p}.
\end{equation}
Applying \textbf{Lemma~\ref{lemma:inequality1}} to the square of the right-hand side of \eqref{throwaway2} yields the result \eqref{throwaway1}.
\end{proof}

This lemma provides an upper bound on the magnitude of the rough integral \[ \int_t^T \psi(X_s^{t, x, \gamma}, \gamma_s) \, d\bm{\zeta}_s \] in terms of the control $\gamma \in \pvarSpaceSimple{\frac{p}{2}}$. We note that the bound is not intended to be sharp.

Having established that the value function may diverge under minimal constraints, as shown in the previous example, we are naturally led to consider how such degeneracies can be addressed. A key observation is that the cost functional must penalize excessively irregular control paths to ensure well-posedness. This motivates the introduction of a \emph{regularizing cost}, which imposes structure on the admissible control space and prevents pathological behavior.

\begin{definition}[Regularizing Cost]
    Let $\mathcal{S} \subseteq \pvarSpaceSimple{\frac{p}{2}}(J, \mathbb{R}^k)$ be a Banach space. A \emph{regularizing cost} on $\mathcal{S}$ is a function
    \begin{equation}
    \label{regularizingCost}
        \beta \colon \Delta_J \times \pvarSpaceSimple{\frac{p}{2}}(J, \mathbb{R}^k) \to [0, +\infty]
    \end{equation}
    satisfying the following conditions:
    \begin{enumerate}
        \item For fixed $(r, t) \in \Delta_J$, the map $\gamma \mapsto \beta_{r,t}(\gamma)$ is continuous.
        \item $\beta_{r,t}(\gamma) = +\infty$ for all $\gamma \notin \mathcal{S}$.
        \item The growth condition holds:
        \[
        \frac{\beta_{r,t}(\gamma)}{\pvar{\gamma}{\frac{p}{2}}{[r,t]}^{2(1+p)}} \to +\infty \quad \text{as} \quad \pvar{\gamma}{\frac{p}{2}}{[r,t]} \to +\infty.
        \]
    \end{enumerate}
\end{definition}

In practice, the form of the regularizing cost depends on the specific phenomenon being modeled. For example, in the setting of \emph{robust stochastic filtering}, we will see that the regularizing cost naturally takes the form of a negative log-likelihood function.

Now, let $\pvarSpaceSimple{0,p}$ denote the closure of smooth paths in $\pvarSpaceSimple{p}$ with respect to the $|| \cdot ||_p$-seminorm. We now redefine the value function by incorporating the regularizing cost:
\begin{equation}
\label{valueWithRegCost}
    V(t,x) \coloneqq \inf_{\gamma \in \pvarSpaceSimple{0,\frac{p}{2}}} \left\{ J(t,x,\gamma) + \beta_{t,T}(\gamma) \right\}.
\end{equation}

With this adjustment, we can show that the value function remains bounded from below under appropriate conditions.

\begin{proposition}
\label{prop:boundedValue}
    Suppose that the functions $f$ and $g$ in \eqref{costFunctional} are bounded below. Then the value function $V(t, x)$ defined in \eqref{valueWithRegCost} is bounded below.
\end{proposition}

\begin{proof}
\normalsize
By \textbf{Propositions~\ref{prop:regularity_RDE}} and \textbf{\ref{ineq5}}, we have the estimate
\[
\left| \int_t^T \psi\, d\bm{\zeta}_s \right| \lesssim 1 + \pvar{\gamma}{\frac{p}{2}}{[t, T]}^{2(1+p)} \lesssim 1 + \frac{\beta_{t,T}(\gamma)}{2},
\]
where the second inequality follows from the growth condition in \textbf{Definition~\ref{regularizingCost}}. Hence,
\[
J(t,x,\gamma) + \beta_{t,T}(\gamma) \ge \int_t^T f\, ds + g(X_T^{t,x,\gamma}) + \frac{\beta_{t,T}(\gamma)}{2} - C,
\]
for some constant $C > 0$. Since $f$ and $g$ are bounded below, the right-hand side is bounded below, which proves the claim.
\end{proof}

\subsection{Dynamic Programming Principle}
\label{subsec:dymanic_programming_principle}

In general, regularizing costs are not assumed to be additive; that is, one does not require $\beta_{s,t} = \beta_{s,u} + \beta_{u,t}$ for all $s < u < t$. However, to establish the \emph{dynamic programming principle} (DPP), it is necessary to identify a class of costs for which such additivity holds. Once this structure is in place, the DPP provides a natural path to deriving the infinitesimal form of the optimal control problem -- namely, the \emph{Hamilton–Jacobi–Bellman} (HJB) equation.

Before turning to the DPP itself, we first justify that the infimum in the value function may be taken over a suitably regular subclass of admissible controls.

\begin{lemma}
\label{valueWithRegCostLemma}
Suppose $\mathcal{S} \subseteq \pvarSpaceSimple{0,\frac{p}{2}}$ contains all smooth functions from $J$ to $\mathbb{R}^k$. Then
\begin{equation}
\label{valueWithRegCostLemmaEquation}
\begin{split}
    V(t,x) & = \inf_{\gamma \in \pvarSpaceSimple{0,\frac{p}{2}}} \left\{ J(t,x,\gamma) + \beta_{t,T}(\gamma) \right\} \\
           & = \inf_{\gamma \in \mathcal{S}} \left\{ J(t,x,\gamma) + \beta_{t,T}(\gamma) \right\}.
\end{split}
\end{equation}
\end{lemma}

\begin{proof}
\normalsize
Fix $\gamma \in \pvarSpaceSimple{0,\frac{p}{2}}$. By assumption, there exists a sequence $(\gamma^n)_{n=1}^\infty \subset \mathcal{S}$ such that \[ \|\gamma^n - \gamma\|_\infty \to 0 \] as $n \to \infty$. Then $\pvar{\gamma^n - \gamma}{\frac{p}{2}}{J} \to 0$ as $n \to \infty$, and continuity of the rough integral and cost functional (see \textbf{Theorem~\ref{theorem:RDETheorem2}}) ensures that the corresponding values $J(t,x,\gamma^n) + \beta_{t,T}(\gamma^n)$ converge to $J(t,x,\gamma) + \beta_{t,T}(\gamma)$. Hence, the infimum over $\pvarSpaceSimple{0,\frac{p}{2}}$ is the same as the infimum over $\mathcal{S}$.
\end{proof}

\begin{remark}
To retain the dynamic programming principle, it is sufficient to work with a class of additive regularizing costs. One such example arises when the control path is induced by an $L^q$-integrable function. Suppose $u \in W^{1,q}$, and define the control path by
\[
\gamma_s^{t,a,u} \coloneqq a + \int_t^s u_y \, dy,
\]
where $W^{1,q}$ is the Sobolev space of functions with one weak derivative in $L^q$. Then the regularizing cost
\[
\beta_{s,r}(\gamma^{t,a,u}) \coloneqq \varepsilon \int_s^r |u_y|^q \, dy
\]
is clearly additive in time by the properties of integration. We will see that regularizing costs of this form are sufficient for our analysis.
\end{remark}

\begin{definition}[Value Function]
\label{FinalVal}
Suppose $u_s \in \mathcal{W}^{1,q}$ for $s \in J$, and define the control path $\gamma^{t,a,u}$ by $d\gamma_s^{t,a,u} = u_s \, ds$ with initial condition $\gamma_t^{t,a,u} = a$. Then the value function is given by
\begin{equation}
\label{finalValueFunction}
    v(t,x,a) \coloneqq \inf_{u \in L^q} \left\{ J(t,x,\gamma^{t,a,u}) + \varepsilon \int_t^T |u_s|^q \, ds \right\}.
\end{equation}
\end{definition}

\noindent
The regularizing cost in the definition above ensures additivity and thereby enables the application of the dynamic programming principle. The proof proceeds in a manner analogous to the classical case; see, for example, \cite{yong}.
\begin{theorem}[Dynamic Programming Principle]
\label{DPP}
Let $X_s^{t,x,a,u} \coloneqq X_s^{t,x,\gamma^{t,a,u}}$. Suppose \[ 1 \leq q < +\infty \] and $r \in [t, T]$. Then the value function satisfies
\begin{equation}
\label{DPPequation}
\begin{split}
    v(t,x,a) = \inf_{u \in L^q} \Bigg\{ & v\big(r, X_r^{t,x,a,u}, \gamma_r^{t,a,u}\big) + \int_t^r f\big(X_s^{t,x,a,u}, \gamma_s^{t,a,u}\big) \, ds \\
    & + \int_t^r \psi\big(X_s^{t,x,a,u}, \gamma_s^{t,a,u}\big) \, d\bm{\zeta}_s + \varepsilon \int_t^r |u_s|^q \, ds \Bigg\}.
\end{split}
\end{equation}
\end{theorem}

\subsection{Generalized Control Problem}
\label{sec:GOC}

We now generalize the value function defined in (\ref{finalValueFunction}). Observe that the regularization cost can be absorbed into the integral involving the function \(f\). Fixing a Banach space \((U, \|\cdot\|_U)\), we reformulate the optimal control problem as follows:
\begin{equation}
\label{reformualedOC}
\begin{split}
    J(t,x,a,u) \coloneqq & \int_t^T f(X^{t,x,a,u}_s, \gamma^{t,a,u}_s) \, ds \\
    & + \int_t^T \psi(X^{t,x,a,u}_s, \gamma^{t,a,u}_s) \, d\bm{\zeta}_s + g(X^{t,x,a,u}_T, \gamma^{t,a,u}_T),
\end{split}
\end{equation}
where the controlled processes evolve according to
\begin{equation}
\label{process_txau}
    dX_s^{t,x,a,u} = b(X_s^{t,x,a,u}, \gamma_s^{t,a,u}) \, ds + \lambda(X_s^{t,x,a,u}, \gamma_s^{t,a,u}) \, d\bm{\zeta}_s,
\end{equation}
\begin{equation}
    d\gamma_s^{t,a,u} = h(\gamma_s^{t,a,u}, u_s) \, ds,
\end{equation}
\begin{equation}
    u \colon J \to U,
\end{equation}
and the value function is given by
\begin{equation}
\label{reformulatedVF}
    v(t,x,a) \coloneqq \inf_{u \in L^{\infty}} J(t,x,a,u).
\end{equation}

The following assumptions, lemma, and corollary will be used in proving the main results in the next section.

\begin{assumption}
\label{assumption}
We impose the following conditions:
\begin{enumerate}
    \item \(b \in \text{Lip}_b\) and \(\lambda, \psi \in C^3_b\),
    
    \item \(f(x, a, u)\) and \(g(x, a)\) are continuous, bounded below, and Lipschitz continuous in \((x,a)\); moreover, \(f\) is uniformly continuous in \(u\),
    
    \item \(h(a, u)\) is continuous, Lipschitz in \(a\), and uniformly continuous in \(u\); it is also bounded in \(a\), locally uniformly in \(u\), and satisfies, for some \(\delta \ge 1\),
    \[
        \sup_{a \in \mathbb{R}^k} \frac{|h(a,u)|}{\|u\|_U^\delta} \to 0 \quad \text{as } \|u\|_U \to +\infty,
    \]
    
    \item with the same \(\delta\) as above,
    \[
        \inf_{x \in \mathbb{R}^m,\, a \in \mathbb{R}^k} \frac{|f(x,a,u)|}{\|u\|_U^{2\delta(1+p)}} \to +\infty \quad \text{as } \|u\|_U \to +\infty.
    \]
\end{enumerate}
\end{assumption}

\begin{lemma}
\label{corollaryLemma}
    Assume \textbf{Assumption~\ref{assumption}} holds. Then there exists a constant \( 0 < C < +\infty \) such that
    \begin{equation}
        \left| \int_t^T \psi(X_s^{t,x,a,u}, \gamma_s^{t,a,u}) \, d\bm{\zeta}_s \right| \le C + \frac{1}{2} \int_t^T f(X_s^{t,x,a,u}, \gamma_s^{t,a,u}, u_s) \, ds.
    \end{equation}
\end{lemma}

\begin{proof}
\normalsize
    Recall that the function \( p \mapsto \pvar{\gamma}{p}{J} \) is non-increasing for \( 1 \le p < +\infty \). By \textbf{Lemma~\ref{ineq5}} and Hölder’s inequality, we obtain
    \begin{equation*}
    \begin{split}
        \left| \int_t^T \psi(X_s^{t,x,a,u}, \gamma_s^{t,a,u}) \, d\bm{\zeta}_s \right| 
        & \le C\left( 1 + \pvar{\gamma^{t,a,u}}{\frac{p}{2}}{[t,T]}^{2(1+p)} \right) \\
        & \le C\left( 1 + \pvar{\gamma^{t,a,u}}{1}{[t,T]}^{2(1+p)} \right) \\
        & = C\left( 1 + \left[ \int_t^T |h(\gamma_s^{t,a,u}, u_s)| \, ds \right]^{2(1+p)} \right) \\
        & \le C \left( 1 + T^{\frac{2(1+p)}{p'}} \int_t^T |h(\gamma_s^{t,a,u}, u_s)|^{2(1+p)} \, ds \right),
    \end{split}
    \end{equation*}
    where \( p' \) is the Hölder conjugate of \( 2(1+p) \) and we have used that \( T - t \le T \). The claim then follows directly from \textbf{Assumption~\ref{assumption}}.
\end{proof}

\begin{corLemma}
\label{corollary}
    Suppose \( K \subseteq \bb{R}^m \times \bb{R}^k \) is compact. Then there exists \( M > 0 \) such that for all \( (t,x,a) \in J \times K \), the admissible controls \( u \in U \) may be restricted to those satisfying
    \[
        \pvar{\gamma^{t,a,u}}{\frac{p}{2}}{J} \le M.
    \]
\end{corLemma}

\begin{proof}
\normalsize
    From \textbf{Lemma~\ref{corollaryLemma}}, we have the estimate
    \[
        J(t,x,a,u) \ge \frac{1}{2} \int_t^T f(X_s^{t,x,a,u}, \gamma_s^{t,a,u}, u_s) \, ds - \Tilde{C}
    \]
    for some constant \( \Tilde{C} > 0 \). Fix any control \( u^* \in U \). We may then discard any control \( u \) for which
    \[
        \frac{1}{2} \int_t^T f(X_s^{t,x,a,u}, \gamma_s^{t,a,u}, u_s) \, ds - \Tilde{C} \ge \sup_{(t^*,x^*,a^*)} J(t^*,x^*,a^*,u^*).
    \]
    In other words, controls yielding such large values can be excluded from consideration. As the proof of \textbf{Lemma~\ref{corollaryLemma}} provides an upper bound on \( \pvar{\gamma^{t,a,u}}{\frac{p}{2}}{J} \) in terms of the cost, the desired bound follows from \textbf{Assumption~\ref{assumption}} and the arbitrariness of \( u^* \).
\end{proof}

\subsection{Rough HJB Equation and Solutions}
\label{roughHJBsection}
We now derive the rough Hamilton-Jacobi-Bellman (HJB) equation, which governs the value function in a rough path optimal control problem. In classical control theory, the HJB equation arises as a dynamic programming principle applied to the value function. When the system is driven by a rough path, however, the formulation and analysis become more delicate. To handle the irregularity of the driving signal, we begin by approximating the rough path with a sequence of smooth paths. This allows us to use classical results to derive the HJB equation for the smooth setting and then pass to the limit, obtaining a well-defined equation in the rough path sense.

Fix a geometric rough path $\roughpath{\zeta} \in \mathscr{D}^p$, and consider a sequence of smooth paths $(\eta^n_t)_{n = 1}^\infty$ such that $\eta_t^n \to \zeta_t$ uniformly as $n \to +\infty$. We lift each $\eta^n_t$ to a rough path by defining
\begin{equation*}
    (\eta^n)^{(2)}_{s,t} = \int_s^t \eta^n_{s,u} \otimes d\eta^n_u,
\end{equation*}
and set the lifted path as $\bm{\eta}^n = (\eta^n, (\eta^n)^{(2)})$.

Now consider the controlled differential equation driven by the smooth path $\bm{\eta}^n$:
\begin{equation}
\label{smoothHJB}
    dX_s^{t,x,a,u,n} = b(X_s^{t,x,a,u,n}, \gamma_s^{t,a,u})ds + \lambda(X_s^{t,x,a,u,n}, \gamma_s^{t,a,u}) d\bm{\eta}_s^n,
\end{equation}
where the integral is understood in the Riemann-Stieltjes sense. For each $n$, we define the associated controlled rough path solution as
\begin{equation}
    \left( X^{t,x,a,u,n}, (X^{t,x,a,u,n})^\prime \right) = \left( X^{t,x,a,u,n}, \lambda(X^{t,x,a,u,n}, \gamma^{t,a,u}) \right).
\end{equation}
Then, as $n \to +\infty$, we recover the rough differential equation (\ref{process_txau}) in the sense of \textbf{Theorem~\ref{theorem:RDETheorem1}}.

The rough HJB equation is obtained by first solving the classical HJB equation corresponding to each smooth path $\eta^n$ using standard results (e.g., from \cite{bardi}) and then passing to the limit as $n \to +\infty$. We now state the result formally.

\begin{theorem}[Rough HJB Equation]
\label{roughHJB}
    Suppose the setup of equations \textnormal{(\ref{reformualedOC})--(\ref{reformulatedVF})} holds as given at the start of \textnormal{\S\ref{sec:GOC}}. Then the value function $v$ satisfies the rough HJB equation
    \begin{equation}
    \label{roughHJBequation}
        -dv - b \cdot \nabla_x v \, dt - \inf_{u \in U} \left\{ h \cdot \nabla_a v + f \right\} dt - \left( \lambda \cdot \nabla_x v + \psi \right) \, d\bm{\zeta} = 0,
    \end{equation}
    subject to the terminal condition
    \begin{equation}
    \label{roughHJBterminalCondition}
        v(T,x,a) = g(x,a),
    \end{equation}
    where $\nabla_y$ denotes the gradient with respect to the variable $y$.
\end{theorem}

\begin{proof}
\normalsize
    Explicitly, we interpret the rough HJB equation \textnormal{(\ref{roughHJBequation})} as the limiting form of the sequence of classical HJB equations:
    \begin{equation}
        -dv - b \cdot \nabla_x v \, dt - \inf_{u \in U} \left\{ h \cdot \nabla_a v + f \right\} dt - \left( \lambda \cdot \nabla_x v + \psi \right) \, d\bm{\eta}^n = 0.
    \end{equation}
\end{proof}

Before showing that the function $v$ from \textbf{Theorem~\ref{roughHJB}} is the unique viscosity solution of \textnormal{(\ref{roughHJBequation})}, we introduce the appropriate notion of convergence.

\begin{definition}[Viscosity Solution]
\label{viscosity}
    Suppose that $v^{\bm{\eta}^n} \to v^{\bm{\zeta}}$ as $n \to +\infty$ in the sense of \textbf{Theorem~\ref{roughHJB}}. Then $v^{\bm{\zeta}}$ is said to solve the rough HJB equation \textnormal{(\ref{roughHJBequation})} in the viscosity sense if the convergence $v^{\bm{\eta}^n} \to v^{\bm{\zeta}}$ is locally uniform and the driving signals $\bm{\eta}^n$ converge to $\bm{\zeta}$ in the $\frac{1}{p}$-Hölder rough path metric $\varrho_{\frac{1}{p}-\text{Höl}}$.
\end{definition}

We are now ready to state the central result of this section. The following theorem establishes that the value function introduced in the generalized control formulation is the unique viscosity solution to the rough Hamilton-Jacobi-Bellman (HJB) equation. It also shows that this value function depends continuously on the driving rough path, which is essential for the stability and robustness of the control problem under model perturbations. This result extends classical control theory to the rough path setting and forms the theoretical foundation for rough optimal control.
\begin{theorem}[Uniqueness]
\label{mainOC}
    Suppose \textbf{Assumption~\ref{assumption}} holds. Then the value function defined in \textnormal{(\ref{reformulatedVF})} is a viscosity solution of the rough HJB equation \textnormal{(\ref{roughHJBequation})} with terminal condition \textnormal{(\ref{roughHJBterminalCondition})}, in the sense of \textbf{Definition~\ref{viscosity}}.

    Moreover, for each fixed \((t,x,a)\), the map
    \[
        \bm{\zeta} \mapsto v^{\bm{\zeta}}(t,x,a)
    \]
    is uniformly continuous with respect to both the \(\frac{1}{p}\)-Hölder and \(p\)-variation rough path metrics \(\varrho_{\text{\(\frac{1}{p}\)-Höl}, J}\) and \(\varrho_{p, J}\), where \(\bm{\zeta} \in \mathscr{C}^{0,p}_g(J, \mathbb{R}^d)\).
\end{theorem}

\begin{proof}
\normalsize
    Let \(\bm{\eta} \in \mathscr{C}^{0,p}_g(J, \mathbb{R}^d)\) be another geometric rough path such that
    \[
        \left\|\bm{\zeta}\right\|_{\frac{1}{p}\text{-Höl},J} \le M 
        \quad \text{and} \quad 
        \left\|\bm{\eta}\right\|_{\frac{1}{p}\text{-Höl},J} \le M
    \]
    for some fixed \(M > 0\). Then the hypotheses of \textbf{Theorem~\ref{theorem:RDETheorem2}} are satisfied. Define the corresponding controlled paths
    \[
        X^{\bm{\zeta}}_s \coloneqq X^{t,x,a,u,\bm{\zeta}}_s, \qquad X^{\bm{\eta}}_s \coloneqq X^{t,x,a,u,\bm{\eta}}_s.
    \]
    By \textbf{Corollary~\ref{corollary}}, we may restrict to controls \(u\) such that \(\pvar{\gamma^{t,a,u}}{\frac{p}{2}}{J} \le L\) for some constant \(L > 0\). Applying \textbf{Theorem~\ref{theorem:RDETheorem2}} yields
    \[
        \pvar{X^{\bm{\zeta}} - X^{\bm{\eta}}}{\infty}{J} \lesssim \roughMetricPvar{\zeta}{\eta}{p}{J}
    \]
    and
    \[
        \left\| \int_0^\cdot \psi(X^{\bm{\zeta}}_s, \gamma^{t,a,u}_s) d\bm{\zeta}_s - \int_0^\cdot \psi(X^{\bm{\eta}}_s, \gamma^{t,a,u}_s) d\bm{\eta}_s \right\|_{p,J} \lesssim \roughMetricPvar{\zeta}{\eta}{p}{J}.
    \]

    Let \(U^L\) denote the set of all controls \(u\) such that \(d\gamma_s^{t,a,u} = u_s ds\) and \(\pvar{\gamma}{\frac{p}{2}}{J} \le L\). Then for any \((t,x,a)\),
    \begin{equation*}
    \begin{split}
        |v^{\bm{\zeta}}(t,x,a) - v^{\bm{\eta}}(t,x,a)| 
        & \le \sup_{u \in U^L} \Bigg| \int_t^T \Big( f(X^{\bm{\zeta}}_s, \gamma_s, u_s) - f(X^{\bm{\eta}}_s, \gamma_s, u_s) \Big) ds \\
        & \quad + \int_t^T \psi(X^{\bm{\zeta}}_s, \gamma_s) d\bm{\zeta}_s - \int_t^T \psi(X^{\bm{\eta}}_s, \gamma_s) d\bm{\eta}_s \\
        & \quad + g(X_T^{\bm{\zeta}}, \gamma_T) - g(X_T^{\bm{\eta}}, \gamma_T) \Bigg| \\
        & \lesssim \sup_{u \in U^L} \left( \int_t^T |X^{\bm{\zeta}}_s - X^{\bm{\eta}}_s| ds + \roughMetricPvar{\zeta}{\eta}{p}{J} + |X_T^{\bm{\zeta}} - X_T^{\bm{\eta}}| \right) \\
        & \lesssim \roughMetricPvar{\zeta}{\eta}{p}{J} \\ 
        & \lesssim \roughMetricHolder{\zeta}{\eta}{p}{J},
    \end{split}
    \end{equation*}
    where we used the Lipschitz continuity of \(f\) and \(g\).

    Finally, taking a sequence \((\bm{\eta}^n)\) of smooth paths such that
    \[
        \lim_{n \to +\infty} \roughMetricHolder{\zeta}{\eta^n}{p}{J} = 0
    \]
    and applying the above bound completes the proof.
\end{proof}

Additionally, we present a key result known as the \emph{verification theorem}, which asserts that if one can find a pair of functions \( w \) and \( \gamma^* \) satisfying the rough HJB equation, then \( w \) coincides with the \emph{value function} and \( \gamma^* \) is the optimal control. The proof follows a similar strategy to that of \textbf{Theorem~\ref{roughHJB}}, using approximation by classical control problems.
\begin{theorem}[Verification Theorem]
\label{theorem:verification_theorem}
    Suppose the setup in equations \textnormal{(\ref{reformualedOC})--(\ref{reformulatedVF})} holds as specified at the beginning of \textnormal{\S\ref{sec:GOC}}. If the function \( w \) satisfies the rough HJB equation \textnormal{(\ref{roughHJBequation})--(\ref{roughHJBterminalCondition})} with associated control \( \gamma^* \), then \( w \) is the unique \emph{value function} and \( \gamma^* \) is the \emph{optimal control}.
\end{theorem}
\begin{proof}
\normalsize
    As in \textbf{Theorem~\ref{roughHJB}}, we approximate the rough HJB equation by a sequence of control problems driven by smooth paths \( \bm{\zeta}^n \). If $w^n$ and $(\gamma^*)^n$ satisify the classical HJB equation with respect to $\bm{\zeta}^n$, then $w^n$ is the unique \emph{value function} and $(\gamma^*)^n$ is the unique \emph{optimal control}. By \textbf{Theorem~\ref{mainOC}}, we may let $n \to \infty$ to obtain the desired result.
\end{proof}

We have thus established that three classical results -- the \emph{dynamic programming principle}, the \emph{HJB equation}, and the \emph{verification theorem} -- remain valid in the setting of geometric rough paths with \( p \in [2, 3) \). This provides a rigorous foundation for addressing stochastic filtering problems in a rough path framework.

\section{Pathwise Robust Filtering}
\label{sec:pathwise_robust_filtering}
Stochastic filtering is a fundamental mathematical technique for inferring the evolving state of an unobserved or partially observed ``signal process'' from noisy observations. More precisely, we are interested in estimating the state of a signal process governed by the stochastic differential equation  
\[
    dS_t = \alpha_t S_t \,dt + \sigma_t \,dB_t^1,
\]
where \(S_t\) denotes the signal at time \(t\), \(\alpha_t\) is a time-dependent growth rate, \(\sigma_t\) is the volatility, and \(B_t^1\) is a Brownian motion introducing randomness into the system.

Simultaneously, we observe a process \(Y_t\), which is related to the signal through the equation  
\[
    dY_t = c_t S_t \,dt + dB_t^2,
\]
where \(c_t\) modulates the influence of the signal on the observations, and \(B_t^2\) is another Brownian motion representing observation noise. The challenge lies in estimating the unobserved signal \(S_t\) using only the observed data \(Y_t\).

This estimation is complicated by the fact that the parameters \(\alpha_t\) and \(\sigma_t\) are typically unknown or imprecisely estimated in practice. This parameter uncertainty is a central obstacle in classical filtering approaches and motivates the development of methods that are less sensitive to such inaccuracies.

The filtering problem has been extensively studied -- see, for example, the comprehensive treatment by Bain and Crisan \cite{bain}. A wide range of techniques has been developed to estimate both the hidden state and the model parameters, with applications spanning finance, engineering, and environmental science.

In this section, we focus on \emph{robust filtering}, a framework designed to yield reliable estimates even in the presence of parameter uncertainty. The key idea is to penalize poor parameter fits in a principled way, thus promoting more accurate and stable state estimates derived from observed data. Robust filtering is especially valuable when model parameters are misspecified or vary with time.

Finally, we show that the robust filtering problem admits a natural reformulation as a pathwise optimal control problem. This perspective is powerful, as it allows us to apply the tools developed in Section \S\ref{OC} to construct effective and theoretically sound filtering methods.

\subsection{The Kalman-Bucy Filter}
\label{subsec:kalman-bucy_filter}
To ground our discussion of robust filtering, we begin with the classical Kalman-Bucy filter, which provides a solution to the filtering problem when both the signal and observation processes are linear and driven by Gaussian noise. This model serves as a foundational example and will inform the structure of our robust, pathwise formulation.

The signal and observation processes are given by the stochastic differential equations
\begin{equation}
\label{signal}
    dS_t = \alpha_t S_t \,dt + \sigma_t \,dB_t^1,
\end{equation}
\begin{equation}
\label{observation}
    dY_t = c_t S_t \,dt + dB_t^2,
\end{equation}
where \eqref{signal} describes the unobserved \emph{signal} process \(S_t\), and \eqref{observation} models the observed process \(Y_t\). We assume \(S_t\) takes values in \(\mathbb{R}^m\), \(Y_t\) in \(\mathbb{R}^d\), and that \(Y_0 = 0\). The initial condition for the signal is \(S_0 \sim \mathcal{N}(\mu_0, \Sigma_0)\). The coefficients are given by time-dependent functions:
\begin{itemize}
    \item \(\alpha \colon J \to \mathbb{R}^{m \times m}\) (drift matrix),
    \item \(\sigma \colon J \to \mathbb{R}^{m \times l}\) (diffusion matrix),
    \item \(c \colon J \to \mathbb{R}^{d \times m}\) (observation matrix).
\end{itemize}

We allow for correlation between the Brownian motions \(B_t^1\) and \(B_t^2\), and assume their quadratic covariation satisfies
\[
    d\langle B^1, B^2 \rangle_t = \rho_t \,dt,
\]
where \(\rho_t \in \mathbb{R}^{l \times d}\). We further assume that
\[
    I - \rho_t \rho_t^\top
\]
is positive semi-definite for all \(t\), ensuring that the correlation structure is well-defined.

Let \(\mathcal{Y}_t\) denote the completed filtration generated by the observation process \((Y_s)_{s \le t}\). The conditional expectation
\[
    q_t = \mathbb{E}[S_t \mid \mathcal{Y}_t]
\]
provides the optimal estimate of the signal given the observations up to time \(t\). The Kalman-Bucy filter asserts that \(q_t\) satisfies the stochastic differential equation
\begin{equation}
    dq_t = \alpha_t q_t \,dt + \left( R_t c_t^\top + \sigma_t \rho_t \right) \left( dY_t - c_t q_t \,dt \right),
\end{equation}
where \(R_t\), the conditional covariance matrix
\[
    R_t = \mathbb{E}[(S_t - q_t)(S_t - q_t)^\top \mid \mathcal{Y}_t],
\]
evolves according to the Riccati differential equation
\begin{equation}
    \frac{dR_t}{dt} = \sigma_t \sigma_t^\top + \alpha_t R_t + R_t \alpha_t^\top - \left( R_t c_t^\top + \sigma_t \rho_t \right) \left( c_t R_t + \rho_t^\top \sigma_t^\top \right).
\end{equation}

This classical setting illustrates the structure and potential of filtering theory in the linear-Gaussian case. In the next section, we move beyond this framework to develop a pathwise and robust formulation that accommodates model uncertainty.

\subsection{Robust Filtering}
\label{subsec:robust_filtering}
Building on the Kalman-Bucy framework, we now introduce the machinery required to model filtering under uncertainty -- what we refer to as \emph{robust filtering}. In this formulation, the key idea is to replace fixed model parameters with dynamically controlled variables, allowing us to hedge against potential misspecification. This leads naturally to a control-theoretic interpretation of the filtering problem.

Let \(\gamma_t \coloneqq (\alpha_t, \sigma_t, c_t, \rho_t) \in \Gamma\) denote a vector of control variables, where the control space is defined as
\[
    \Gamma \coloneqq \mathbb{R}^{m \times m} \times \mathbb{R}^{m \times l} \times \mathbb{R}^{d \times m} \times \Upsilon,
\]
and
\[
    \Upsilon \coloneqq \left\{ \rho_t \in \mathbb{R}^{l \times d} \colon I - \rho_t \rho_t^\top \ \text{is positive definite} \right\}.
\]
Here, \(\alpha_t\), \(\sigma_t\), \(c_t\), and \(\rho_t\) are viewed as control inputs, rather than fixed coefficients, allowing the filter to adaptively respond to uncertainty in the system.

To quantify belief about the signal \(S_t\) given observations up to time \(t\), we introduce a convex expectation operator that penalizes poor model choices. Given a bounded Borel function \(\varphi \colon \mathbb{R}^m \to \mathbb{R}\), the \emph{robust conditional expectation} is defined as
\begin{equation}
\label{convexExpectation}
    \mathcal{E}(\varphi(S_t) \mid \mathcal{Y}_t) \coloneqq \esssup_{(\gamma, \mu_0, \Sigma_0)} \left\{ \mathbb{E}^{\gamma, \mu_0, \Sigma_0} \left[ \varphi(S_t) \mid \mathcal{Y}_t \right] - \left( \frac{1}{k_1} \beta(\gamma, \mu_0, \Sigma_0 \mid \mathcal{Y}_t) \right)^{k_2} \right\},
\end{equation}
where \(k_1 > 0\) and \(k_2 \ge 1\) control the degree of penalization. The essential supremum is taken over admissible model parameters \((\gamma, \mu_0, \Sigma_0)\), and the penalty function \(\beta\) encodes a measure of model plausibility, which we now define.

\begin{definition}[Admissible Controls]
\label{admissibleControls}
    Let \(\mathcal{A}\) denote the set of all \emph{admissible controls} \(\gamma_t \in \Gamma\), where \(\gamma \colon J \to \Gamma\) is absolutely continuous with bounded derivative.
\end{definition}

\begin{definition}[Penalty Function]
\label{penaltyDefinition}
    The penalty function \(\beta\) is given by a negative log-likelihood:
    \begin{equation}
    \label{penaltyFunction}
        \beta_t(\gamma, \mu_0, \Sigma_0 \mid \mathcal{Y}_t) = - \ln \left( \pi_t(\gamma, \mu_0, \Sigma_0) L_t(\gamma, \mu_0, \Sigma_0 \mid \mathcal{Y}_t) \right),
    \end{equation}
    where \(\pi_t\) is the prior density and \(L_t\) is the likelihood of the observation path under the parameters \((\gamma, \mu_0, \Sigma_0)\).
\end{definition}

This robust framework naturally yields confidence intervals and point estimates that are resistant to model misspecification.

\begin{remark}[Robust Point Estimate and Confidence Interval]
\label{researchEstimateAndCI}
    Using the convex expectation \(\mathcal{E}\), one can define a \emph{robust point estimate} as the minimizer of the mean squared error:
    \begin{equation}
    \label{robustPointEstimate}
        \argmin_{\xi \in \mathbb{R}} \mathcal{E} \left( (\varphi(S_t) - \xi)^2 \mid \mathcal{Y}_t \right).
    \end{equation}
    Additionally, a robust confidence interval is given by
    \begin{equation}
    \label{robustCI}
        \left[ -\mathcal{E} \left( -\varphi(S_t) \mid \mathcal{Y}_t \right),\ \mathcal{E} \left( \varphi(S_t) \mid \mathcal{Y}_t \right) \right],
    \end{equation}
    where the endpoints reflect the range of plausible values under model uncertainty.
\end{remark}

To proceed further with the formulation of the robust filtering problem, we impose a structural assumption on the prior density:

\begin{assumption}
\label{logPriorAssumption}
    \emph{We assume that the negative log-prior density admits the representation}
    \begin{equation}
    \label{logPriorEquation}
        -\ln (\pi_t(\gamma, \mu_0, \Sigma_0)) = \int_0^t z(q_s, R_s, \gamma_s) \, ds + g(\mu_0, \Sigma_0),
    \end{equation}
    \emph{where} \(z\) \emph{is a measurable function and} \(g\) \emph{captures dependence on the initial conditions.}
\end{assumption}

Next, we turn to the log-likelihood component in the penalty function \eqref{penaltyFunction}. It can be expressed as a Radon-Nikodym derivative:
\begin{equation}
\label{radonNyk}
    L_t(\gamma, \mu_0, \Sigma_0 \mid \mathcal{Y}_t) = \left( \frac{d\mathbb{P}^{\gamma, \mu_0, \Sigma_0}}{d\mathbb{P}^{\gamma^*, \mu_0^*, \Sigma_0^*}} \right)_{\mathcal{Y}_t},
\end{equation}
where \((\gamma^*, \mu_0^*, \Sigma_0^*)\) denotes a fixed set of reference parameters. It turns out that this likelihood can be computed explicitly using the innovation process.

Recall from \cite{bain} that the \emph{innovation process} is given by
\begin{equation}
\label{innovProcess}
    dV_s = dY_s - c_s q_s \, ds,
\end{equation}
which is a \(\mathcal{Y}_t\)-adapted Brownian motion under the measure \(\mathbb{P}^{\gamma, \mu_0, \Sigma_0}\). Furthermore, the conditional mean and innovation process under the reference model, denoted \(q^*_s\) and \(V^*_s\), satisfy
\begin{equation}
\label{conditionalInnov}
    dV_s = dV^*_s - (c_s q_s - c_s^* q^*_s) \, ds.
\end{equation}
Applying Girsanov’s theorem \cite{oksendal}, the likelihood is then given by
\begin{equation}
    L_t(\gamma, \mu_0, \Sigma_0 \mid \mathcal{Y}_t)
    = \exp \left( \int_0^t (c_s q_s - c_s^* q^*_s) \cdot dV^*_s
    - \frac{1}{2} \int_0^t |c_s q_s - c_s^* q^*_s|^2 \, ds \right).
\end{equation}

Substituting the definition \(dV^*_s = dY_s - c_s^* q_s^* ds\) into the expression above yields
\begin{equation}
\label{logExpression}
\begin{split}
    -\ln L_t(\gamma, \mu_0, \Sigma_0 \mid \mathcal{Y}_t)
    &= - \int_0^t (c_s q_s - c_s^* q^*_s) \cdot dY_s
    + \frac{1}{2} \int_0^t \left( |c_s q_s|^2 - |c_s^* q^*_s|^2 \right) ds \\
    &= - \int_0^t c_s q_s \cdot dY_s
    + \frac{1}{2} \int_0^t |c_s q_s|^2 \, ds + \text{const},
\end{split}
\end{equation}
where the last equality follows from treating the reference parameters as fixed, so their contribution may be absorbed into a constant.

In what follows, we disregard the additive constant and consider the penalty function to be defined up to an arbitrary constant offset.

To prepare for the reformulation of the robust filtering problem as a pathwise optimal control problem, we now convert the Itô integral appearing in equation \eqref{logExpression} into its Stratonovich form. Specifically, we write
\begin{equation}
\label{stratonovich}
    - \int_0^t c_s q_s \cdot dY_s = - \int_0^t c_s q_s \circ dY_s + \frac{1}{2} \langle cq, Y \rangle_t,
\end{equation}
where $\circ$ denotes the Stratonovich integral. As shown in \cite{allan2019pathwise}, the quadratic covariation term satisfies
\begin{equation}
\label{quadVar}
    \langle cq, Y \rangle_t = \int_0^t \text{trace} \left( c_s ( R_s c_s^\top + \sigma_s \rho_s ) \right) ds.
\end{equation}

Substituting \eqref{stratonovich} and \eqref{quadVar} into \eqref{logExpression}, we obtain the following expression for the negative log-likelihood:
\begin{equation}
    - \ln L_t(\gamma, \mu_0, \Sigma_0 | \mathcal{Y}_t) = - \int_0^t c_s q_s \circ dY_s + \frac{1}{2} \int_0^t \left( |c_s q_s|^2 + \text{trace} \left( c_s ( R_s c_s^\top + \sigma_s \rho_s ) \right) \right) ds.
\end{equation}

To streamline the notation, we define
\begin{equation*}
    w(q, R, \gamma) \coloneqq z(q, R, \gamma) + \frac{1}{2} \left( |cq|^2 + \text{trace} \left( c(Rc^\top + \sigma \rho) \right) \right),
\end{equation*}
and
\begin{equation*}
    \psi(q, \gamma) \coloneqq -cq,
\end{equation*}
so that the convex expectation \eqref{convexExpectation} becomes
\begin{equation}
\label{finalConvex}
\begin{split}
    \mathcal{E}(\varphi(S_t) | \mathcal{Y}_t) & = \esssup_{\gamma, \mu_0, \Sigma_0} \bigg\{ \bb{E} \left[ \varphi(S_t) \mid \mathcal{Y}_t \right] \\
    & \quad - \left( \frac{1}{k_1} \left( \int_0^t w(q_s, R_s, \gamma_s) ds + \int_0^t \psi(q_s, \gamma_s) \circ dY_s + g(\mu_0, \Sigma_0) \right) \right)^{k_2} \bigg\}.
\end{split}
\end{equation}

This representation sets the stage for the pathwise optimal control formulation introduced in the following section.

\subsection{Lifting into Rough Path Space}
\label{subsec:lifting_into_rough_path_space}

In practice, filtering is typically carried out with respect to a fixed realization of the observation path. That is, given a specific sample path \(\zeta_t = Y_t(\omega)\), one seeks to solve the filtering problem pathwise. However, this requires lifting the path \(\zeta_t\) into rough path space.

We define the \emph{lift} of \(\zeta_t\) as
\begin{equation}
    \lift{\zeta}{s,t} = \int_s^t Y_{s,r}(\omega) \otimes \circ dY_r(\omega),
\end{equation}
so that \(\roughpath{\zeta} \in \mathscr{C}_g^{0,p}\) for \(p \in (2,3)\).

With this lift, the prediction process \(q_t\) satisfies the rough differential equation
\begin{equation}
    dq_t = \alpha_t q_t dt + \left( R_t c_t^\top + \sigma_t \rho_t \right) \left( d\zeta_t - c_t q_t dt \right).
\end{equation}

Moreover, using \textbf{Theorem \ref{theorem:rough_integration}}, we obtain the increment representation
\begin{equation}
\begin{split}
    q_{s,t} & = \int_s^t \left( R_r c_r^\top + \sigma_r \rho_r \right) d\zeta_r + O(|t-s|) \\
    & = \left( R_r c_r^\top + \sigma_r \rho_r \right) \zeta_{s,t} + O(|t-s|).
\end{split}
\end{equation}

This implies that the Gubinelli derivative of \(\psi(q, \gamma) = -cq\) is given by
\[
    \psi(q, \gamma)^\prime = -c \left( R_r c_r^\top + \sigma_r \rho_r \right).
\]

Consequently, the rough integral
\begin{equation*}
    \int_0^\cdot \psi(q_s, \gamma_s) d\bm{\zeta}_s
\end{equation*}
exists and coincides with the Stratonovich integral.

\subsection{The Optimal Control Problem}
\label{subsec:the_optimal_control_problem}

We are now ready to recast the filtering problem as a pathwise optimal control problem. Define the functional
\begin{equation}
\label{functional}
    k_t(\mu, \Sigma) \coloneqq \inf \left\{ \int_0^t w(q_s, R_s, \gamma_s) ds + \int_0^t \psi(q_s, \gamma_s) d\bm{\zeta}_s + g(q_0, R_0) \right\},
\end{equation}
where the infimum is taken over all trajectories \(\gamma, q_0, R_0\) such that \((q_t, R_t) = (\mu, \Sigma)\). In addition, we set \(g(\mu_0, \Sigma_0) = +\infty\) for all \((\mu_0, \Sigma_0) \notin \bb{R}^m \times \mathcal{S}_+^m\), where \(\mathcal{S}_+^m\) denotes the space of symmetric, positive-definite \(m \times m\) matrices.

This formulation allows us to express the convex expectation in (\ref{finalConvex}) as follows:

\begin{lemma}
\label{ceLemma}
Let \(\phi(\cdot | \mu, \Sigma)\) denote the probability density function of the \(N(\mu, \Sigma)\) distribution. If \(\roughpath{\zeta}\) is defined as above, then for any bounded measurable function \(\varphi\), we have
\begin{equation}
\label{ceLemmaEquation}
    \ce{S}{t} = \sup \left\{ \int_{\bb{R}^m} \varphi(x) d\phi(x | \mu, \Sigma) - \left( \frac{1}{k_1} k_t(\mu, \Sigma) \right)^{k_2} \right\},
\end{equation}
where the supremum is taken over all \((\mu, \Sigma) \in \bb{R}^m \times \mathcal{S}_+^m\).
\end{lemma}

A proof of \textbf{Lemma \ref{ceLemma}} can be found in \cite{ac}. Thus, for any pair \((q^{t, \mu, \Sigma}, R^{t, \mu, \Sigma})\) satisfying the terminal condition \((q_t, R_t) = (\mu, \Sigma)\), the filtering problem reduces to the optimal control problem
\begin{equation}
\label{filteringOC}
    k_t(\mu, \Sigma) = \inf_\gamma \left\{ \int_0^t w(q_s, R_s, \gamma_s) ds + \int_0^t \psi(q_s, \gamma_s) d\bm{\zeta}_s + g(q_0, R_0) \right\}.
\end{equation}

Note, however, that this formulation lacks a regularizing cost, which is necessary to avoid degeneracy. To address this, we introduce control dynamics of the form
\begin{equation*}
    d\gamma_s^{t,a,u} = h(\gamma_s^{t,a,u}, u_s)ds,
\end{equation*}
where
\begin{equation*}
    h \colon \Gamma \times U \to U, \quad u \colon J \to U \ \text{is bounded},
\end{equation*}
and \(U \coloneqq \bb{R}^{m \times m} \times \bb{R}^{m \times l} \times \bb{R}^{d \times m} \times \bb{R}^{l \times d}\).

We also impose a new terminal condition:
\begin{equation}
\label{terminalCondition}
    (q_t^{t, \mu, \Sigma, a, u}, R_t^{t, \mu, \Sigma, a, u}, \gamma_t^{t, a, u}) = (\mu, \Sigma, a).
\end{equation}

Furthermore, we allow the functions \(w\) and \(g\) to depend on the initial control \(\gamma_0\) without affecting the result of \textbf{Lemma \ref{ceLemma}}, as shown in \cite{allan2019pathwise}. Define the revised functional
\begin{equation}
\label{newFunctional}
    \Tilde{k}_t(\mu, \Sigma) \coloneqq \inf_{a \in \Gamma} v(t, \mu, \Sigma, a),
\end{equation}
where the value function \(v\) is given by
\begin{equation*}
\begin{split}
    v(t, \mu, \Sigma, a) & \coloneqq \inf_{u \text{ bounded}} \bigg\{ \int_0^t w(q_s, R_s, \gamma_s, u_s) ds + \int_0^t \psi(q_s, \gamma_s) d\bm{\zeta}_s \\
    & \qquad + g(q_0, R_0, \gamma_0) \bigg\},
\end{split}
\end{equation*}
with all processes indexed as \(q_s = q_s^{t, \mu, \Sigma, a, u}\), \(R_s = R_s^{t, \mu, \Sigma, a, u}\), and \(\gamma_s = \gamma_s^{t, a, u}\).

This formulation defines a fully regularized, pathwise optimal control problem for robust filtering.

\subsection{The Associated HJB Equation}
\label{subsec:the_associated_hjb_equation}
We are now in a position to derive the rough Hamilton–Jacobi–Bellman (HJB) equation corresponding to the robust filtering problem. For notational simplicity, we introduce the following system of equations:
\begin{equation*}
    dq_s^{t, \mu, \Sigma, a, u} = b_\mu(q_s^{t, \mu, \Sigma, a, u}, R_s^{t, \Sigma, a, u}, \gamma_s^{t, a, u})\,ds + \lambda(R_s^{t, \Sigma, a, u}, \gamma_s^{t, a, u})\,d\zeta_s, \quad q_t^{t, \mu, \Sigma, a, u} = \mu,
\end{equation*}
\begin{equation*}
    dR_s^{t, \Sigma, a, u} = b_\Sigma(R_s^{t, \Sigma, a, u}, \gamma_s^{t, a, u})\,ds, \quad R_t^{t, \Sigma, a, u} = \Sigma,
\end{equation*}
\begin{equation*}
    d\gamma_s^{t, a, u} = h(\gamma_s^{t, a, u}, u_s)\,ds, \quad \gamma_t^{t, a, u} = a,
\end{equation*}
where \(\gamma = (\alpha, \sigma, c, \rho)\) is the control. The drift and diffusion terms are defined as
\begin{align*}
    b_\mu(q, R, \gamma) &\coloneqq \alpha q - \left( R c^\top + \sigma \rho \right) c q, \\
    b_\Sigma(R, \gamma) &\coloneqq \sigma \sigma^\top + \alpha R + R \alpha^\top - \left( R c^\top + \sigma \rho \right) \left( c R + \rho^\top \sigma^\top \right), \\
    \lambda(R, \gamma) &\coloneqq R c^\top + \sigma \rho.
\end{align*}

\begin{remark}[Backward Control Problem]
The system above defines a backward control problem: it is initialized at the terminal time \(t\), with the cost functional \(g(q_0, R_0, \gamma_0)\) acting on the initial states. In order to apply the results from Section~\S\ref{OC}, this backward problem must be appropriately reformulated.
\end{remark}

We now define some notation and state the regularity conditions required for the derivation of the associated HJB equation. First, let
\[
    \|A\| \coloneqq \text{trace}(A^\top A), \quad \|\gamma\| \coloneqq \max\{ \|\alpha\|, \|\sigma\|, \|c\|, \|\rho\| \},
\]
and for \(A \in \mathcal{S}_+^m\), let \(\lambda_{\min}(A)\) and \(\lambda_{\max}(A)\) denote its minimum and maximum eigenvalues, respectively.

We now impose the following assumptions:

\begin{assumption}
\label{filteringAssumptions}
\
\begin{itemize}
    \item The functions \(w(q, R, \gamma, u)\) and \(g(q, R, \gamma)\) are continuous, bounded from below, and locally Lipschitz in \((q, R, \gamma)\), uniformly in \(u\).
    
    \item The function \(h(\gamma, u)\) is continuous; satisfies \(\{ h(\gamma, u) \colon u \in U \} = U\) for all \(\gamma \in \Gamma\); is Lipschitz in \(\gamma\), uniformly in \(u\); is bounded in \(\gamma\), locally uniformly in \(u\); and for some \(\delta_1 > 0\),
    \[
        \sup_{\gamma \in \Gamma} \frac{\|h(\gamma, u)\|}{\|u\|^{\delta_1}} \to 0 \quad \text{as } \|u\| \to \infty.
    \]
    
    \item There exists \(\delta_2 > \delta_1\) such that
    \[
        \frac{|f(q, R, \gamma, u)|}{\left(1 + |q| + \|R\|^2 + \|\gamma\|^2\right) \|u\|^{\delta_2} + \left(1 + |q|^2 + \|R\|^2\right)(1 + \|\gamma\|^4)} \to \infty
    \]
    as \(|q| + \|R\| + \|\gamma\| + \|u\| \to \infty\).
    
    \item The terminal cost \(g\) satisfies:
    \[
        \frac{|g(q, R, \gamma)|}{|q|^2 + (1 + \|R\|)(1 + \|\gamma\|^2)} \to \infty \quad \text{as } |q| + \|R\| + \|\gamma\| \to \infty,
    \]
    and
    \[
        \inf_{(q, \gamma) \in \bb{R}^m \times \Gamma} |g(q, R, \gamma)| \to \infty \quad \text{as } \lambda_{\min}(R) \to 0.
    \]
    
    \item Furthermore, \(g\) diverges as the observation noise degenerates:
    \[
        \inf_{q, R, \alpha, \sigma, c} |g(q, R, \gamma)| \to \infty \quad \text{as } \lambda_{\max}(\rho \rho^\top) \to 1.
    \]
    
    \item The control growth is bounded by the observation structure:
    \[
        \|h(\gamma, u)\| \le (1 - \lambda_{\max}(\rho \rho^\top)) \|u\| \quad \text{for all } (\gamma, u) \in \Gamma \times U.
    \]
\end{itemize}
\end{assumption}

The condition \(\{ h(\gamma, u) \colon u \in U \} = U\) ensures that for any given terminal condition \((t, \mu, \Sigma, a)\), there exists an admissible control such that the state trajectories remain in their respective domains. This is essential for the well-posedness and finiteness of the value function \(v\), as shown in \cite{allan2019pathwise}.

The next two results mirror \textbf{Lemma~\ref{corollaryLemma}} and \textbf{Corollary~\ref{corollary}}, and serve as crucial steps in establishing the main result of this section -- namely, an analog of \textbf{Theorem~\ref{mainOC}} for the filtering setting.

\begin{lemma}
\label{lemmaFilt}
Suppose \textbf{Assumption~\ref{filteringAssumptions}} holds. Then, for any terminal condition \((t, \mu, \Sigma, a)\) and control \(u\), the following estimate holds:
\begin{equation*}
    \left| \int_0^t \psi(q_s, \gamma_s)\,d\bm{\zeta}_s \right| \le C + \frac{1}{2} \left( \int_0^t w(q_s, R_s, \gamma_s)\,ds + g(q_0, R_0, \gamma_0) \right),
\end{equation*}
for some constant \(C > 0\).
\end{lemma}

The proof is technical and relies on the simplified notation introduced at the beginning of this section, as well as on \textbf{Assumption~\ref{filteringAssumptions}}. The key idea is to apply \textbf{Theorem~\ref{theorem:rough_integration}} and to derive bounds involving \(|\psi(q, \gamma)|\), \(\| \psi(q, \gamma)' \|\), and the \(p\)-variation norms \(\pvar{R^{\psi(q, \gamma)}}{p/2}{[0, t]}\) and \(\pvar{\psi(q, \gamma)'}{p}{[0, t]}\). For full details, see \cite{allan2019pathwise}.

\begin{corLemma}
\label{filteringCorollary}
Suppose \(K \subseteq \bb{R}^m \times \mathcal{S}_+^m \times \Gamma\) is compact. Then, for terminal conditions \((t, \mu, \Sigma, a) \in [0, T] \times K\), the optimal control problem can be restricted to controls \(u\) such that
\[
    \|q\|_\infty, \|R\|_\infty, \|\gamma\|_\infty, \pvar{R}{1}{[0, t]}, \pvar{\gamma}{1}{[0, t]} \le M
\]
for some constant \(M > 0\).
\end{corLemma}

\begin{proof}
\normalsize
The bound on \(\pvar{\gamma}{1}{[0,t]}\) follows by arguments similar to those in \textbf{Corollary~\ref{corollaryLemma}}. The process \(R_t\) remains in a bounded set due to its ODE and \textbf{Assumption~\ref{filteringAssumptions}}, implying \(\|R\|_\infty < \infty\) and thus \(\pvar{R}{1}{[0, t]} < \infty\). Finally, boundedness of \(\|q\|_\infty\) follows directly from the differential equation governing \(q\).
\end{proof}

We now proceed to derive the HJB equation. As in \S\ref{roughHJBsection}, we approximate the rough path \(\bm{\zeta}\) by a sequence of smooth rough paths \(\bm{\eta}^n = (\eta^n, (\eta^n)^{(2)})\) via the Stone–Weierstrass theorem, solve the control problem in the smooth setting, and take the limit to obtain the rough HJB equation.

Before stating the result, let \(A : B\) denote the inner product in the appropriate inner product space. For matrices \(A, B\), this is given by \(A : B \coloneqq \text{trace}(A^\top B)\).

\begin{theorem}[Rough HJB Equation for Filtering]
\label{filteringHJBtheorem}
Suppose \textbf{Assumption~\ref{filteringAssumptions}} holds. Then the value function \(v\) satisfies the following rough HJB equation:
\begin{equation}
\label{filteringHJBequation}
    dv + \left( b_\mu \cdot \nabla_\mu v + b_\Sigma : \nabla_\Sigma v \right) dt + \sup_{u \in U} \left\{ h : \nabla_a v - w \right\} dt + \left( \lambda \cdot \nabla_\mu v - \psi \right) d\bm{\zeta} = 0,
\end{equation}
with terminal condition
\begin{equation}
\label{terminalConditionHJB}
    v(0, \mu, \Sigma, a) = g(\mu, \Sigma, a).
\end{equation}
\end{theorem}

To ensure uniqueness of the solution to (\ref{filteringHJBequation}), one typically restricts attention to the class \(\mathcal{H}\) of value functions \(\Tilde{v}(t, \mu, \Sigma, a)\) that diverge to \(\pm\infty\) as any of the following occur:
\begin{itemize}
    \item \(|\mu| + \|\Sigma\| + \|a\| \to \infty\),
    \item \(\lambda_{\min}(\Sigma) \to 0\),
    \item \(\lambda_{\max}(\rho^\top \rho) \to 1\) when \(\rho\) is random.
\end{itemize}
Justification for this condition is discussed in \cite{ac}, and is beyond the scope of the present work.

We now present a result that is needed to establish the main theorem of this section.

\begin{theorem}[Young Integral]
\label{youngIntegral}
Let \(V\) and \(W\) be Banach spaces, and suppose \(1 \le p, q \le \infty\) satisfy \(\frac{1}{p} + \frac{1}{q} > 1\). If \(X \in \pvarSpace{p}{J}{V}\) and \(Y \in \pvarSpace{q}{J}{\bm{L}(V, W)}\), then for each \(t \in J\),
\begin{equation}
    \int_0^t Y_s\,dX_s = \lim_{|\mathcal{D}| \to 0} \sum_\mathcal{D} Y_{t_i} \left( X_{t_{i+1}} - X_{t_i} \right),
\end{equation}
and
\begin{equation}
\label{youngIneq}
    \pvar{\int_0^\cdot \left( Y_s - Y_0 \right) dX_s}{p}{J} \lesssim \pvar{Y}{q}{J} \pvar{X}{p}{J}.
\end{equation}
\end{theorem}

A proof of this theorem can be found in \cite{lyons}. We conclude this section with the central result.

\begin{theorem}[Uniqueness]
\label{theorem:filtering}
The value function \(v\) defined in \textbf{Theorem~\ref{filteringHJBtheorem}} is a viscosity solution of the HJB equation in the sense of \textbf{Definition~\ref{viscosity}}. Moreover, the map
\[
\bm{\zeta} \mapsto v^{\bm{\zeta}}(t, \mu, \Sigma, a), \quad \bm{\zeta} \in \mathscr{C}^{0,p}_g
\]
is locally uniformly continuous with respect to both rough path metrics \(\roughMetricHolder{\cdot}{\cdot}{p}{J}\) and \(\roughMetricPvar{\cdot}{\cdot}{p}{J}\), locally uniformly in \((t, \mu, \Sigma, a)\).
\end{theorem}

\begin{proof}
\normalsize
Let \(\bm{\eta} \in \mathscr{C}^{0,p}_g\) be another rough path such that, without loss of generality, \(\roughMetricHolder{\zeta}{\eta}{p}{J} \le 1\). Define
\[
M^* \coloneqq \max \left\{ \roughholder{\zeta}{1}{p}{J}, \roughholder{\eta}{1}{p}{J} \right\}.
\]
Let \(q^{\bm{\zeta}}, q^{\bm{\eta}}\) denote the state trajectories driven by \(\bm{\zeta}, \bm{\eta}\), respectively, and similarly define \(v^{\bm{\zeta}}, v^{\bm{\eta}}\). Let the bound \(M\) be as defined in \textbf{Corollary~\ref{filteringCorollary}}, applying to both \(\bm{\zeta}\) and \(\bm{\eta}\).

From inequality (\ref{youngIneq}), we have
\[
\left| \int_s^t \left( R_r c^\top_r + \sigma_r \rho_r \right)\,d(\eta - \zeta)_r \right| \lesssim \pvar{\eta - \zeta}{p}{J},
\]
which implies
\[
|q_s^{\bm{\eta}} - q_s^{\bm{\zeta}}| \lesssim \int_s^t |q_r^{\bm{\eta}} - q_r^{\bm{\zeta}}|\,dr + \pvar{\eta - \zeta}{p}{J}.
\]
By Grönwall’s inequality, this gives
\[
\pvar{q^{\bm{\eta}} - q^{\bm{\zeta}}}{\infty}{J} \lesssim \pvar{\eta - \zeta}{p}{J}.
\]

Applying \textbf{Theorem~\ref{theorem:RDETheorem2}}, we obtain
\[
\pvar{\int_0^\cdot \psi(q_s^{\bm{\eta}}, \gamma_s)\,d\bm{\eta}_s - \int_0^\cdot \psi(q_s^{\bm{\zeta}}, \gamma_s)\,d\bm{\zeta}_s}{p}{J} \lesssim \roughMetricPvar{\eta}{\zeta}{p}{J}.
\]

Now, for any terminal condition \((t, \mu, \Sigma, a) \in J \times K\), we estimate:
\begin{equation*}
\begin{split}
|v^{\bm{\eta}}(t, \mu, \Sigma, a) - v^{\bm{\zeta}}(t, \mu, \Sigma, a)| & \le \sup_{u} \Bigg| \int_0^t \left[ w(q_s^{\bm{\eta}}, R_s, \gamma_s, u_s) - w(q_s^{\bm{\zeta}}, R_s, \gamma_s, u_s) \right]\,ds \\
& \quad + \int_0^t \psi(q_s^{\bm{\eta}}, \gamma_s)\,d\bm{\eta}_s - \int_0^t \psi(q_s^{\bm{\zeta}}, \gamma_s)\,d\bm{\zeta}_s \\
& \quad + g(q_0^{\bm{\eta}}, R_0, \gamma_0) - g(q_0^{\bm{\zeta}}, R_0, \gamma_0) \Bigg| \\
& \lesssim \sup_u \left( \int_0^t |q_s^{\bm{\eta}} - q_s^{\bm{\zeta}}|\,ds + \roughMetricPvar{\eta}{\zeta}{p}{J} + |q_0^{\bm{\eta}} - q_0^{\bm{\zeta}}| \right) \\
& \lesssim \roughMetricPvar{\eta}{\zeta}{p}{J} \\
& \lesssim \roughMetricHolder{\eta}{\zeta}{p}{J}.
\end{split}
\end{equation*}

Here, the supremum is taken over all controls \(u\) such that \(\gamma\) satisfies the required assumptions. The remainder of the proof follows the same structure as that of \textbf{Theorem~\ref{mainOC}}.
\end{proof}

\section{Optimal Stopping with Signatures}
\label{sec:optimal_stopping_with_signatures}

Before exploring optimal stopping using path signatures, it is helpful to begin with a concrete example that illustrates its practical relevance. Continuing our focus on finance, we consider the setting of American option pricing -- a fundamental problem in the theory of financial derivatives.

An American \emph{call} option grants the holder the right, but not the obligation, to purchase an underlying asset at a predetermined price, known as the strike price, at any time before the option's expiration. For example, suppose an agent holds an American call option on a stock with a strike price of \$20. The agent may choose to exercise the option at any point before it expires. If the stock price remains at or below \$20, the option will not be exercised, as doing so would yield no gain. However, if the stock price rises above \$20, the agent can exercise the option, buy the stock at the lower strike price, and sell it at the market price to realize a profit.

Viewed from a different angle, suppose the agent is considering purchasing an American call option with a three-month expiry and a \$20 strike. A natural question arises: what is a fair price for this contract? One widely accepted formulation for determining the fair price is:
\begin{equation}
\label{example:american_option}
V = \sup_\tau \mathbb{E} \left[ \exp(-r \tau) \max(S_\tau - 20, 0) \right],
\end{equation}
where \(\tau\) ranges over all stopping times between the present and expiry, \(S_\tau\) denotes the stock price at time \(\tau\), and \(r\) is the risk-free interest rate. This expression represents the maximum expected discounted payoff from exercising the option, and thus reflects the option's fair value. Intuitively, an agent seeking a positive expected profit should not pay more than this amount for the contract.

More generally, our goal is to study problems of the form (\ref{example:american_option}), but where the underlying process \(S_t\) is modeled as a rough path. This introduces new mathematical challenges, particularly around defining and computing optimal stopping times in a pathwise setting. A key objective of this section is to develop the theoretical and computational tools needed to reformulate such problems in rough path space, and to demonstrate how these tools make the theory of optimal stopping both numerically tractable and practically applicable in the rough path framework.

\subsection{Stopped Rough Paths}
\label{subsec:stopped_rough_paths}

We begin with a definition. To model adaptedness to a filtration, we consider restrictions of rough paths defined on the interval $[0, T]$ to subintervals $[0, s] \subseteq [0, T]$. 

Recall from $\S$\ref{subsec:rough_path_signatures} that $\Hat{\Omega}^p_T$ denotes the closure of the canonical lifts $\Hat{\bb{X}}^{\le \lfloor p \rfloor}$ of piecewise $\bb{R}^d$-valued smooth paths $X$, where $\Hat{X}_t \coloneqq (t, X_t)$ for $0 \le t \le T$. If \[ \Hat{\bb{X}}_t = (t, X_t) \in \Hat{\Omega}_s^p, \] we extend $\Hat{\bb{X}}$ to an element of $\Hat{\Omega}_T^p$ by setting $\Tilde{\bb{X}}_t = (t, X_{s \land t})$. Formally:

\begin{definition}[Stopped Rough Paths]
\label{def:stopped_rough_paths}
    Let $0 < T < \infty$. The \emph{space of stopped rough paths} is defined as
    \begin{equation}
    \label{eq:stoppe_rough_paths}
        \Lambda_T^p \coloneqq \bigcup_{t \in [0, T]} \Hat{\Omega}_t^p,
    \end{equation}
    where each $\Hat{\bb{X}} \in \Hat{\Omega}_t^p$ is extended to an element of $\Hat{\Omega}_T^p$ as described above.

    We equip $\Lambda_T^p$ with the metric
    \begin{equation}
    \label{eq:stopped_rough_path_metric}
        d(\hat{\bb{X}}|_{[0,t]}, \hat{\bb{Y}}|_{[0,s]}) \coloneqq d_{p\text{-var};[0,t]}(\bb{X}|_{[0,t]}, \Tilde{\bb{Y}}|_{[0,t]}) + |t - s|,
    \end{equation}
    where we assume $s \le t$.
\end{definition}

\subsection{Randomized Stopping Times}
\label{subsec:randomized_stopping_times}

Since we will reformulate the optimal stopping problem in terms of rough path signatures, we require a lemma showing that stopping times with respect to the filtration generated by a rough path can be represented as measurable functions of the rough path.

The proof of the following lemma can be found in the appendix:
\begin{lemma}
\label{lemma:indicator}
    Let $\Hat{\bb{X}}$ be a stochastic process in $\Hat{\Omega}^p_T$, and let $\bb{F} \coloneqq (\bb{F}_t)_{0 \le t \le T}$ be the filtration generated by $\Hat{\bb{X}}$. If $\tau$ is a stopping time with respect to $\bb{F}$, then there exists a Borel measurable function $\theta \colon \Lambda_T^p \to \{0, 1\}$ such that $\theta(\Hat{\bb{X}}(\omega)|_{[0,t]}) = \mathds{1}_{\{\tau(\omega) \le t\}}$.
\end{lemma}

Let us now introduce some notation. Throughout the following, suppose that $(\Omega, \mathcal{F}, P)$ is a probability space such that:
\begin{enumerate}
    \item $\Hat{\bb{X}}$ is a stochastic process in $\Hat{\Omega}_T^p$;
    \item $\bb{F} = (\bb{F}_t)$ is the filtration given by $\bb{F}_t \coloneqq \sigma(\Hat{\bb{X}}|_{[0,s]} : 0 \le s \le t) \subseteq \mathcal{F}$;
    \item $Y \colon [0,T] \times \Omega \to \bb{R}$ is a continuous-time stochastic process adapted to $\bb{F}$;
    \item $\mathcal{S}$ denotes the set of all $\bb{F}$-stopping times.
\end{enumerate}

The goal is to solve an optimal stopping problem for the process $Y$, where the stopping time is defined in terms of the rough path $\Hat{\bb{X}}$.

\begin{definition}[Randomized Stopping Times]
\label{def:randomized_stopping_times}
    Let $\mathcal{T} \coloneqq C(\Lambda_T^p,\bb{R})$ denote the set of all continuous stopping policies, and let $Z$ be a non-negative random variable independent of $\Hat{\bb{X}}$, with $P(Z = 0) = 0$. Given $\theta \in \mathcal{T}$, define the \emph{randomized stopping time}
    \begin{equation}
    \label{eq:def:randomized_stopping_time}
        \tau_\theta^r \coloneqq \inf \left\{ t \ge 0 \colon \int_0^{t \land T} \theta(\Hat{\bb{X}}|_{[0,s]})^2 \, ds \ge Z \right\},
    \end{equation}
    with the convention that $\inf \emptyset = +\infty$.
\end{definition}

We square the policy $\theta$ in equation \eqref{eq:def:randomized_stopping_time} to ensure the integrand is non-negative. To show that classical stopping times can be approximated by randomized ones, we need the following result, which can be found in \cite{wisn}.

\begin{lemma}[Continuous Approximation]
\label{lemma:continuous_approximation}
    Let $\mu$ be a finite Borel measure on a metric space $X$, and let $f \colon X \to \bb{R}$ be a $\mu$-measurable function. Then there exists a sequence of continuous functions $f_n \colon X \to \bb{R}$ such that $f_n \to f$ almost everywhere.
\end{lemma}

We now establish the approximation result:

\begin{proposition}[Stopping Time Approximation]
\label{prop:stopping_time_approximation}
    For every $\tau \in \mathcal{S}$, there exists a sequence $\theta_n \in \mathcal{T}$, $n \in \bb{N}$, such that $\tau_{\theta_n}^r \to \tau$ almost surely. Moreover, if $\bb{E}||Y||_\infty < \infty$, then
    \begin{equation}
    \label{eq:proposition:EY_randomized}
        \sup_{\theta \in \mathcal{T}} \bb{E}[Y_{\tau_\theta^r \land T}] = \sup_{\tau \in \mathcal{S}} \bb{E}[Y_{\tau \land T}].
    \end{equation}
\end{proposition}

\begin{proof}
\normalsize
    Choose a sequence $\Tilde{\theta}_n \in \mathcal{T}$ such that $\Tilde{\theta}_n(\Hat{\bb{X}}|_{[0,t]}) \to \mathds{1}_{\{ \tau \le t \}}$ almost surely with respect to $\lambda \otimes P$, where $\lambda$ is Lebesgue measure. Without loss of generality, assume $0 \le \Tilde{\theta}_n \le 1$. Define $\theta_n = (2\Tilde{\theta}_n)^n$. Then
    \begin{equation*}
        \lim_{n \to \infty} \theta_n(\Hat{\bb{X}}|_{[0,t]}) = 
        \begin{dcases}
            +\infty & \text{if } t \ge \tau, \\
            0 & \text{if } t < \tau.
        \end{dcases}
    \end{equation*}
    Hence, $\tau_{\theta_n}^r \to \tau$ almost surely as $n \to \infty$.

    By the dominated convergence theorem,
    \[
        \sup_{\theta \in \mathcal{T}} \bb{E}[Y_{\tau_\theta^r \land T}] \ge \sup_{\tau \in \mathcal{S}} \bb{E}[Y_{\tau \land T}].
    \]

    To prove the reverse inequality, suppose $\theta \in \mathcal{T}$. Since $Z$ is independent of $\Hat{\bb{X}}$, we have
    \[
        \bb{E} [Y_{\tau_{\theta}^r \land T} | \Hat{\bb{X}}] = \int_0^\infty Y_{\tau_z \land T} \, dP_Z(z),
    \]
    where $P_Z$ is the law of $Z$ and
    \[
        \tau_z \coloneqq \inf \left\{ t \ge 0 \colon \int_0^{t \land T} \theta(\Hat{\bb{X}}|_{[0,s]})^2 \, ds \ge z \right\}.
    \]
    Taking expectations yields
    \begin{equation*}
    \begin{split}
        \bb{E}[Y_{\tau_{\theta}^r \land T}] 
        &= \bb{E}\left[\bb{E}[Y_{\tau_{\theta}^r \land T} \mid \Hat{\bb{X}}]\right] \\
        &= \int_0^\infty \bb{E}[Y_{\tau_z \land T}] \, dP_Z(z) \\
        &\le \sup_{\tau \in \mathcal{S}} \bb{E}[Y_{\tau \land T}].
    \end{split}
    \end{equation*}
\end{proof}

It turns out that the conditional expectation $\bb{E}[Y_{\tau^r_\theta \land S} \mid \Hat{\bb{X}}]$ can be expressed in a convenient integral form; the proof is provided in the appendix:

\begin{proposition}[Integral Representation]
\label{proposition:integral_representation}
    Suppose $S$ is an $\bb{F}$-stopping time. If $F_Z$ denotes the cumulative distribution function of $Z$, then
    \begin{equation*}
    \begin{split}
        \bb{E}[Y_{\tau^r_\theta \land S} \mid \Hat{\bb{X}}] &= \int_0^S Y_t \, d\Tilde{F}(t) + Y_S(1 - \Tilde{F}(S)) \\
        &= \int_0^S (1 - \Tilde{F}(t)) \, dY_t + Y_0,
    \end{split}
    \end{equation*}
    where
    \[
        \Tilde{F}(t) \coloneqq F_Z \left( \int_0^t \theta(\Hat{\bb{X}}|_{[0,s]})^2 \, ds \right).
    \]
\end{proposition}

\subsection{Linear Signature Stopping Policies}
\label{subsec:linear_signature_stopping_policies} 

In this section, we introduce and analyze linear signature stopping policies. We conclude by presenting the main theorem, which shows that the original optimal stopping problem can be effectively addressed using such policies.

\begin{definition}[Linear Signature Stopping Policies]
\label{def:linear_signature_stopping_policies}
    Define the space of \emph{linear signature stopping policies} as
    \begin{equation}
    \label{eq:def:linear_signature_stopping_policies}
        \mathcal{T}_{\mathrm{sig}} \coloneqq \left\{ \theta \in \mathcal{T} \colon \exists\, l \in T((\bb{R}^{1+d})^*)\ \text{such that } \theta(\Hat{\bb{X}}|_{[0,t]}) = \langle l, \Hat{\bb{X}}_{0,t}^{<\infty} \rangle,\ \forall\, \Hat{\bb{X}}|_{[0,t]} \in \Lambda_T^p \right\}.
    \end{equation}
\end{definition}

\begin{remark}[Induced Randomized Stopping Time]
\label{remark:induced_randomized_stopping_time}
    Every $l \in T((\bb{R}^{1+d})^*)$ defines a randomized stopping time via the stopping policy $\theta_l(\Hat{\bb{X}}|_{[0,t]}) \coloneqq \langle l, \Hat{\bb{X}}_{0,t}^{<\infty} \rangle$. Specifically,
    \begin{equation*}
    \begin{split}
        \tau_l^r & \coloneqq \tau_{\theta_l}^r \\
        &= \inf \left\{ t \ge 0 \colon \int_0^{t \land T} \langle l, \Hat{\bb{X}}_{0,s}^{<\infty} \rangle^2 \, ds \ge Z \right\}.
    \end{split}
    \end{equation*}
\end{remark}

The following lemma, which can be found in \cite{kalsi}, establishes the density of linear signature stopping policies in a suitable sense.

\begin{lemma}
\label{lemma:Tsig_dense}
    Let $P$ be a probability measure on $(\Hat{\Omega}_T^p, \mathcal{B}(\Hat{\Omega}_T^p))$, and fix $\varepsilon > 0$. Then there exists a compact set $\mathcal{K} \subset \Hat{\Omega}_T^p$ such that $P(\mathcal{K}) > 1 - \varepsilon$, and for every $\theta \in \mathcal{T}$, there exists a sequence $\theta_n \in \mathcal{T}_{\mathrm{sig}}$ satisfying
    \begin{equation*}
        \sup_{\Hat{\bb{X}} \in \mathcal{K},\ t \in [0,T]} \left| \theta_n(\Hat{\bb{X}}|_{[0,t]}) - \theta(\Hat{\bb{X}}|_{[0,t]}) \right| \to 0
    \end{equation*}
    as $n \to \infty$. That is, the class $\mathcal{T}_{\mathrm{sig}}$ is dense in $\mathcal{T}$ when restricted to the compact set $\mathcal{K}$.
\end{lemma}

Recall \textbf{Proposition \ref{prop:stopping_time_approximation}}, which established that stopping times can be approximated by randomized stopping policies. We now show that, under suitable conditions, this approximation can be achieved using linear signature stopping policies.

\begin{theorem}[Linear Signature Stopping Time Approximation]
\label{theorem:linear_signature_stopping_time_approximation}
    Suppose that $Z$ has a continuous density $h$ and that $\bb{E}||Y||_\infty < \infty$. Then,
    \begin{equation}
    \label{eq:linear_signature_stopping_time_approximation}
    \begin{split}
        \sup_{l \in T((\bb{R}^{1+d})^*)} \bb{E}[ Y_{\tau_l^r \land T} ] &= \sup_{\theta \in \mathcal{T}} \bb{E}[ Y_{\tau_\theta^r \land T} ] \\
        &= \sup_{\theta \in \mathcal{T}_{\mathrm{sig}}} \bb{E}[ Y_{\tau_\theta^r \land T} ].
    \end{split}
    \end{equation}
\end{theorem}

\begin{proof}
\normalsize
    It suffices to prove that
    \[
        \sup_{\theta \in \mathcal{T}} \bb{E}[ Y_{\tau_\theta^r \land T} ] \le \sup_{\theta \in \mathcal{T}_{\mathrm{sig}}} \bb{E}[ Y_{\tau_\theta^r \land T} ],
    \]
    since the reverse inequality holds trivially by the definition of $\mathcal{T}_{\mathrm{sig}}$. Fix $\theta \in \mathcal{T}$. By \textbf{Lemma \ref{lemma:Tsig_dense}}, for any $\varepsilon > 0$, there exists a compact set $\mathcal{K}$ with $P(\mathcal{K}) > 1 - \varepsilon$, and a sequence $\theta_n \in \mathcal{T}_{\mathrm{sig}}$ such that
    \begin{equation}
    \label{eq:limiting_case}
        \lim_{n \to \infty} \sup_{\bb{X} \in \mathcal{K},\, t \in [0,T]} \left| \theta_n(\bb{X}|_{[0,t]}) - \theta(\bb{X}|_{[0,t]}) \right| = 0.
    \end{equation}

    Define the functions
    \[
        \Tilde{F}_n(t) \coloneqq F_Z\left( \int_0^t \theta_n(\Hat{\bb{X}}|_{[0,s]})^2\, ds \right), \quad \Tilde{F}(t) \coloneqq F_Z\left( \int_0^t \theta(\Hat{\bb{X}}|_{[0,s]})^2\, ds \right).
    \]
    Since $F_Z$ is continuous and hence uniformly continuous on compact sets, we have
    \begin{equation*}
    \begin{split}
        \left| \bb{E}[Y_T(1 - \Tilde{F}_n(T)); A] - \bb{E}[Y_T(1 - \Tilde{F}(T)); A] \right| &\le \bb{E}\left[ |Y_T| \cdot |\Tilde{F}_n(T) - \Tilde{F}(T)|; A \right] \\
        &\le \bb{E}[|Y_T|] \cdot \sup_{\bb{X} \in \mathcal{K}} |\Tilde{F}_n(T) - \Tilde{F}(T)| ]] \\
        & \to 0 \\
    \end{split}
    \end{equation*}
    as $n \to \infty$.

    Next, observe that
    \[
        \sup_{n \in \bb{N}} \sup_{\bb{X} \in \mathcal{K},\, t \in [0,T]} \left| \theta_n(\bb{X}|_{[0,t]}) \right| < \infty,
    \]
    and since the values lie in a compact set, the convergence
    \[
        \lim_{n \to \infty} \sup_{\bb{X} \in \mathcal{K},\, t \in [0,T]} \left| \theta_n(\bb{X}|_{[0,t]})^2 - \theta(\bb{X}|_{[0,t]})^2 \right| = 0
    \]
    follows by dominated convergence. Consequently,
    \[
        \lim_{n \to \infty} \sup_{\bb{X} \in \mathcal{K}} \left| \int_0^T \theta_n(\bb{X}|_{[0,t]})^2\, ds - \int_0^T \theta(\bb{X}|_{[0,t]})^2\, ds \right| = 0.
    \]

    For the complement $A^c = \mathcal{K}^c$, we simply note that
    \[
        \left| \bb{E}[Y_T(1 - \Tilde{F}_n(T)); A^c] - \bb{E}[Y_T(1 - \Tilde{F}(T)); A^c] \right| \le 2 \bb{E}|Y_T|,
    \]
    which can be made arbitrarily small by choosing $\mathcal{K}$ appropriately.

    Finally, we compare the integral terms:
    \begin{equation*}
    \begin{split}
        &\left| \bb{E} \left[ \int_0^T Y_t\, d\Tilde{F}_n(t) - \int_0^T Y_t\, d\Tilde{F}(t) \right] \right| \\
        &= \Bigg| \bb{E} \left[ \int_0^T Y_t\, \theta_n(\Hat{\bb{X}}|_{[0,t]})^2\, h\left( \int_0^t \theta_n(\Hat{\bb{X}}|_{[0,s]})^2 ds \right) dt \right] \\
        &\quad - \bb{E} \left[ \int_0^T Y_t\, \theta(\Hat{\bb{X}}|_{[0,t]})^2\, h\left( \int_0^t \theta(\Hat{\bb{X}}|_{[0,s]})^2 ds \right) dt \right] \Bigg| \to 0, \\
    \end{split}
    \end{equation*}
    as $n \to \infty$, by the continuity of $h$ and dominated convergence. This completes the proof.
\end{proof}

We conclude by showing that randomized stopping times were merely a technical device for approximation and are ultimately unnecessary. In fact, the original optimal stopping problem can be reformulated entirely in terms of deterministic stopping rules based on path signatures.

Let $l \in T((\bb{R}^{1+d})^*)$. We define the corresponding \emph{hitting time} as
\begin{equation}
\label{eq:def:orthogonal_hitting_time}
    \tau_l \coloneqq \inf \left\{ t \in [0, T] \colon \langle l, \Hat{\bb{X}}_{0,t}^{<\infty} \rangle \ge 1 \right\}.
\end{equation}

This leads to our main result of this section, which establishes that the optimal stopping problem can be solved using linear signature stopping policies without the need for randomization.

\begin{theorem}
\label{theorem:main_theorem}
    Suppose $\bb{E}[||Y||_{\infty}] < \infty$. Then,
    \begin{equation}
    \label{eq:theorem:signature_hitting_time_approximation}
        \sup_{l \in T((\bb{R}^{1+d})^*)} \bb{E}[Y_{\tau_l \land T}] = \sup_{\tau \in \mathcal{S}} \bb{E}[Y_{\tau \land T}].
    \end{equation}
\end{theorem}

\begin{proof}
\normalsize
    By \textbf{Proposition \ref{prop:stopping_time_approximation}} and \textbf{Theorem \ref{theorem:linear_signature_stopping_time_approximation}}, we already know that
    \[
        \sup_{\tau \in \mathcal{S}} \bb{E}[Y_{\tau \land T}] = \sup_{l \in T((\bb{R}^{1+d})^*)} \bb{E}[Y_{\tau_l^r \land T}].
    \]
    It remains to show that
    \[
        \sup_{l \in T((\bb{R}^{1+d})^*)} \bb{E}[Y_{\tau_l^r \land T}] \le \sup_{l \in T((\bb{R}^{1+d})^*)} \bb{E}[Y_{\tau_l \land T}].
    \]
    Fix $l \in T((\bb{R}^{1+d})^*)$. By the integral representation of randomized stopping times, we have
    \[
        \bb{E}[Y_{\tau_l^r \land T} \mid \Hat{\bb{X}}] = \int_0^\infty Y_{\tau_z \land T} \, dP_Z(z),
    \]
    where
    \[
        \tau_z \coloneqq \inf\left\{ t \ge 0 \colon \int_0^{t \land T} \langle l, \Hat{\bb{X}}_{0,s}^{<\infty} \rangle^2\, ds \ge z \right\}.
    \]
    Taking expectations gives
    \begin{equation*}
    \begin{split}
        \bb{E}[Y_{\tau_l^r \land T}] &= \int_0^\infty \bb{E}[Y_{\tau_z \land T}]\, dP_Z(z) \\
        &\le \sup_{l \in T((\bb{R}^{1+d})^*)} \bb{E}[Y_{\tau_l \land T}],
    \end{split}
    \end{equation*}
    since $\tau_z \ge \tau_l$ for all $z$, and thus the randomized stopping rule yields no greater expected reward than the deterministic one. This completes the proof.
\end{proof}

\subsection{Linear Optimal Stopping Problem}
\label{subsec:linearizing_the_optimal_stopping_problem}

Assume now that \( Z \sim \text{Expo}(1) \). Combining \textbf{Theorem \ref{theorem:main_theorem}} with \textbf{Proposition \ref{proposition:integral_representation}}, we obtain the following expression:
\begin{equation}
\label{equation:linear}
\begin{split}
    \sup_{\tau \in \mathcal{S}} \bb{E}[Y_{\tau \land T}] &= \sup_{l \in T((\bb{R}^{1+d})^*)} \bb{E}[Y_{\tau_l^r \land T}] \\
    &= \sup_{l \in T((\bb{R}^{1+d})^*)} \bb{E} \left[ \int_0^T \exp\left( -\int_0^t \langle l, \Hat{\bb{X}}^{< \infty}_{0,s} \rangle^2 ds \right) dY_t \right] + \bb{E}[Y_0],
\end{split}
\end{equation}
which can be approximated using Monte Carlo methods, replacing the infinite signature with a truncated version \(\Hat{\bb{X}}^{\le N}\).

Interestingly, there is a more efficient way to approximate this expression. The key insight is the identity
\[
\int_0^t \langle l, \Hat{\bb{X}}^{< \infty}_{0,s} \rangle^2 ds = \langle (l \shuffle l) \mathbf{1}, \Hat{\bb{X}}^{<\infty}_{0,t} \rangle,
\]
which converts the integral of a quadratic form into a linear functional on the signature. See \cite{bayer2021optimalstoppingsignatures} for further details. With additional tools, one can transform the integral in \eqref{equation:linear} into a form where the integrand is linear, thus eliminating the nonlinearity introduced by the exponential function.

We begin by introducing a key definition:

\begin{definition}[Exponential Shuffle]
\label{def:exponential_shuffle}
    Let \( l \in T(V^*) \), and write \( l = a_0 \emptyset + \Tilde{l} \), where \( \langle \Tilde{l}, \mathbf{1} \rangle = 0 \). The \emph{exponential shuffle} of \( l \) is defined by
    \begin{equation}
    \label{eq:def:exponential_shuffle}
        \exp^\shuffle(l) \coloneqq \exp(a_0) \sum_{r=0}^\infty \frac{1}{r!} \Tilde{l}^{\shuffle r}.
    \end{equation}
\end{definition}

The next result, stated without proof due to its technical nature, provides an approximation bound for replacing the exponential of a signature pairing with the exponential shuffle. We refer the reader to \cite{bayer2021optimalstoppingsignatures} for a complete derivation.

\begin{lemma}
\label{lemma:exponential_shuffle_inequality}
    Let \( l \in T(V^*) \) and \( \mathbf{g} \in G(V) \), where \( G(V) \) is defined as in \eqref{G}. If \( N > 2 |l| \cdot \text{deg}(l) \cdot |\pi_{\text{deg}(l)}(\mathbf{g})| \), then
    \[
    \left| \exp(\langle l, \mathbf{g} \rangle) - \left\langle \exp^\shuffle(l), \pi_{\le N}(\mathbf{g}) \right\rangle \right| \le 4 \exp(\langle l, 1 \rangle) \cdot \frac{\left(|l| \cdot \pi_{\le \text{deg}(l)}(\mathbf{g})\right)^{\lfloor N / \text{deg}(l) \rfloor + 1}}{(\lfloor N / \text{deg}(l) \rfloor + 1)!}.
    \]
\end{lemma}

We now show that the exponential shuffle admits a differential representation, which will be useful for later developments. Note that the proof is deferred to the appendix.

\begin{lemma}
\label{lemma:exponential_differential_equation}
    Let \( l = \lambda_1 w_1 + \cdots + \lambda_n w_n \in T((\bb{R}^{1+d})^*) \). Then
    \[
        \frac{d}{dt} \left\langle \exp^\shuffle(l1), \Hat{\bb{X}}^{\le N}_{0,t} \right\rangle = \sum_{i=1}^n \left\langle \lambda_i w_i, \Hat{\bb{X}}^{< \infty}_{0,t} \right\rangle \left\langle \exp^\shuffle(l1), \Hat{\bb{X}}^{\le N - \deg(w_i) - 1}_{0,t} \right\rangle.
    \]
\end{lemma}

The lemmas established above are instrumental in proving the following result, which we state without proof due to its technical complexity and length; see \cite{bayer2021optimalstoppingsignatures} for full details.

\begin{theorem}[Main Result]
    Suppose \( k > 0 \) and define
    \[
        S_k \coloneqq \inf\{ t \ge 0 \colon \| \Hat{\bb{X}} \|_{p\text{-var}; [0,t]} \ge k \} \land T.
    \]
    If \( Z \sim \text{Expo}(1) \) and \( \bb{E}\|Y\|_\infty < \infty \), then
    \begin{equation*}
    \begin{split}
        \sup_{l \in T((\bb{R}^{1+d})^*)} \bb{E}Y_{\tau^r_l \land T} & = \bb{E}Y_0 \\
        & + \lim_{k \to \infty} \lim_{K \to \infty} \lim_{N \to \infty}
        \sup_{|l| + \deg(l) \le K}
        \bb{E} \left[ \int_0^{S_k} \left\langle \exp^\shuffle \left( -(l \shuffle l)1 \right), \Hat{\bb{X}}^{\le N}_{0,t} \right\rangle dY_t \right] \\
    \end{split}
    \end{equation*}
    where the first two limits may be interchanged.
\end{theorem}

This theorem highlights a key structural insight: the value function of the optimal stopping problem can be approximated by a sequence of expectations involving truncated signatures and linear functionals. In particular, the integral is expressed entirely in terms of the signature of the path and a linear evaluation map, enabling potential numerical approximation through Monte Carlo methods or deep learning-based approaches (see \cite{bayer2021optimalstoppingsignatures} for numerical methods).

However, we emphasize that the development of such numerical methods lies beyond the scope of this paper. Our objective has been purely theoretical: to rigorously characterize the optimal stopping problem within the rough path framework and demonstrate that it can be reduced to a tractable linear optimization problem over signature features.

\section{Discussion and Conclusion}
\label{sec:discussion_and_conclusion}
In this paper, we have provided a coherent and unified exposition of rough path theory and its applications to optimal control, robust filtering, and optimal stopping. Adopting a pathwise and deterministic perspective, we demonstrated how many classical results of stochastic analysis can be reformulated -- and, in certain cases, extended -- within the rough path framework. Our aim was to address the lack of comprehensive and unifying treatments in this area, offering both clarity and mathematical rigor.

To this end, we presented detailed proofs where further elaboration was warranted, such as in \textbf{Proposition~\ref{prop:regularity_RDE}}, and developed a complete formulation of optimal control for the case of $p$-variation with $2 \le p < 3$. Another key contribution is our proof of the \emph{verification theorem} (\textbf{Proposition~\ref{theorem:verification_theorem}}) in this setting -- a central result in control theory. This theorem establishes that if a candidate control and candidate value function satisfy the HJB equation, then they coincide with the unique optimal control and value function. The result formalizes the common practical approach in optimal control, where one posits an ansatz for the solution and verifies its validity via the HJB equation. By extending this methodology to the rough path setting, we have not only bridged a gap in the literature but also provided a rigorous framework for solving optimal control problems in systems driven by rough signals.W This final section outlines several promising directions for further research that build upon the contributions made herein.

\subsection{Optimal Control}

In our study of optimal control within the rough path framework, we demonstrated that a naive attempt to control the noise term in the state dynamics leads to degeneracies in the problem formulation. To resolve this, we introduced a regularization term in the cost functional that penalizes excessive variation in the control, thereby restoring well-posedness. This modification enabled the recovery of the \emph{dynamic programming principle}, the \emph{HJB equation}, and the \emph{verification theorem} in the setting of $p$-variation with $2 \le p < 3$. However, these results are currently unavailable for general $p$-variation, particularly when $3 \le p < \infty$. Closing this gap represents a significant theoretical challenge and a natural avenue for future research. A key objective would be to extend both the \emph{HJB equation} and the \emph{verification theorem} to the broader class of $p$-variations, thereby generalizing the control theory developed in this work.

\subsection{Robust Filtering}
In our treatment of robust stochastic filtering, we approached the problem through the lens of pathwise control, drawing on recent developments that recast filtering as a deterministic optimization problem. One open question, previously noted in \cite{allan2019pathwise}, concerns the convergence behavior of the convex expectation \(\ce{S}{t}\) to the conditional expectation \(\mathbb{E} \left[ \varphi(S_t) \mid \mathcal{Y}_t \right]\). Establishing convergence properties in a robust, pathwise setting is crucial for understanding the reliability and interpretability of the filter.

Further research could also focus on the empirical performance of the pathwise robust filter. Investigating its behavior under various sources of model uncertainty and observation noise, possibly through numerical experiments or applications to real-world data, would enhance our understanding of its practical value and robustness.

\subsection{Optimal Stopping}
While our treatment of optimal stopping under rough paths is mathematical, there remains considerable scope for developing numerical procedures tailored to this setting. Improved algorithms -- particularly those capable of handling high-dimensional or irregular signals -- could lead to greater computational efficiency and wider applicability. Validation through numerical experiments, alongside comparisons with classical stochastic methods, would be instrumental in assessing the practical strengths and limitations of the rough path approach.

\subsection{Final Remarks}
Collectively, the directions outlined above suggest a rich landscape for future exploration. By deepening our understanding of optimal control, filtering, and stopping in the rough path setting, we not only extend classical stochastic theory but also contribute to a broader shift toward deterministic, pathwise formulations of uncertainty. The results presented in this paper lay the groundwork for these developments and invite further theoretical and applied investigations connecting rough analysis with stochastic decision theory.

\newpage
\begin{appendices}

\section{Proofs}
\label{sec:Proofs}

\begin{proof}[Proof of Proposition \ref{prop:regularity_RDE}]
\label{proof:indicator}
\normalsize
    The results above hold for any sub-interval $[s, t]$ of $J~=~[0, T]$ so we will restrict ourselves to $[s, t]$. Recall \textbf{Lemmas \ref{lemma:variation_inclusion}, \ref{lemma:1var_has_finite_pvar}, \ref{lemma:inequality1}, \ref{lemma:inequality2}} and \textbf{\ref{lemma:inequality3}}. Since
    \begin{equation*}
    \begin{split}
        \psi(X, \gamma)^{\prime} & = \partial_{x} \psi(X, \gamma) X^{\prime} \\
        & = \partial_{x} \psi(X, \gamma) \lambda(X, \gamma) \\
    \end{split}
    \end{equation*}
    and $\psi, \lambda \in C^2_b$, it follows that that $\partial_{x} \psi, \lambda$ are Lipschitz continuous due to their bounded derivatives. To simplify the notation further, let $\Delta_s^t X \coloneqq X_{s,t}$. Then
    \begin{equation*}
    \begin{split}
        |\Delta_s^t \psi(X, \gamma)^{\prime}| & \lesssim |(X_{s, t}, \gamma_{s, t})| \\
        & \lesssim |X_{s, t}| + |\gamma_{s, t}| \\
        & \lesssim \pvar{X}{p}{[s, t]} + \pvar{\gamma}{\frac{p}{2}}{[s, t]}
    \end{split}
    \end{equation*}
    so $\pvar{\psi}{p}{[s, t]} \lesssim \pvar{X}{p}{[s, t]} + \pvar{\gamma}{\frac{p}{2}}{[s, t]}$, proving (\ref{est1}).
    
    To prove (\ref{est2}) we expand $R^\psi$ using Taylor's theorem, i.e.
    \begin{equation}
    \label{taylorRemainder}
    \begin{split}
        R^{\psi}_{s, t} & = \Delta_s^t \psi(X, \gamma) - \psi(X_s, \gamma_s)^{\prime}(X_{s, t}, \gamma_{s, t}) \\
        & = \frac{1}{2} \partial_x^2 \psi (X_s + h X_{s, t}, \gamma_s) (X_{s, t}, \gamma_{s, t})^{\otimes 2} \\
    \end{split}
    \end{equation}
    for some $h \in [0, 1]$. Before proceeding, note that $p \mapsto \pvar{X}{p}{[s, t]}$ is non-increasing for any path $X$. Hence
    \begin{equation*}
    \begin{split}
        |R^{\psi}_{s, t}| & \lesssim |(X_{s,t}, \gamma_{s, t})|^2 \\
        & \lesssim \pvar{X}{p}{[s, t]}^2 + \pvar{\gamma}{\frac{p}{2}}{[s, t]} \\
        & \lesssim \pvar{X}{p}{[s, t]}^2 + \pvar{R^X}{\frac{p}{2}}{[s, t]} + \pvar{\gamma}{\frac{p}{2}}{[s, t]} \\
    \end{split}
    \end{equation*}
    by (\ref{taylorRemainder}).

    Now we prove (\ref{est3}). By \textbf{Theorem~\ref{theorem:rough_integration}} and since $b \in Lip_b$, we have
    \begin{equation*}
    \begin{split}
        |R^X_{s, t}| & = |X_{s,t} - X^{\prime}_{s} \zeta_{s,t}| \\
        & = \left| \int_s^t b(X_u, \gamma_u)du + \int_s^t \lambda(X_u, \gamma_u)d\bm{\zeta}_u - \lambda(X_s, \gamma_s) \zeta_{s, t} \right| \\
        & \le \left| \int_s^t \lambda(X_u, \gamma_u)d\bm{\zeta}_u - \lambda(X_s, \gamma_s) \zeta_{s, t} - \lambda(X_s, \gamma_s)^{\prime} \lift{\zeta}{s,t} \right| \\
        & + \left| \int_s^t b(X_u, \gamma_u)du \right| + \left| \lambda(X_s, \gamma_s)^{\prime} \lift{\zeta}{s,t} \right| \\
        & \lesssim \pvar{R^\lambda}{\frac{p}{2}}{[s, t]} \pvar{\zeta}{p}{[s, t]} + \pvar{\lambda(X, \gamma)^{\prime}}{\frac{p}{2}}{[s, t]} \pvar{\lift{\zeta}{}}{\frac{p}{2}}{[s, t]} \\
        & + |t - s| + \pvar{\lift{\zeta}{}}{\frac{p}{2}}{[s, t]}. \\
    \end{split}
    \end{equation*}
    In light of (\ref{est1}), (\ref{est2}) and the above, it follows that
    \begin{equation}
    \label{progress}
    \begin{split}
        \pvar{R^X}{\frac{p}{2}}{[s, t]} & \lesssim \left( \pvar{X}{p}{[s, t]}^2 + \pvar{R^X}{\frac{p}{2}}{[s, t]} + \pvar{\gamma}{\frac{p}{2}}{[s, t]} \right) \pvar{\zeta}{p}{[s, t]} \\
        & + \left( \pvar{X}{p}{[s, t]} + \pvar{\gamma}{\frac{p}{2}}{[s, t]} \right) \pvar{\lift{\zeta}{}}{\frac{p}{2}}{[s, t]} \\
        & + |t - s| + \pvar{\lift{\zeta}{}}{\frac{p}{2}}{[s, t]} \\
        & \lesssim \left( \pvar{X}{p}{[s, t]}^2 + \pvar{R^X}{\frac{p}{2}}{[s, t]} + \pvar{\gamma}{\frac{p}{2}}{[s, t]} \right) \pvar{\zeta}{p}{[s, t]} \\
        & + \left( 1 + \pvar{X}{p}{[s, t]}^2 + \pvar{\gamma}{\frac{p}{2}}{[s, t]} \right) \pvar{\lift{\zeta}{}}{\frac{p}{2}}{[s, t]} \\
        & + |t - s| + \pvar{\lift{\zeta}{}}{\frac{p}{2}}{[s, t]}. \\
    \end{split}
    \end{equation}
    Looking at the proof of \textbf{Proposition~\ref{prop:rough_finite_pvar}}, we see that $\pvar{\zeta}{p}{[s, t]} \le M |t - s|^{\frac{1}{p}}$, and similarly for $\lift{\zeta}{}$. Also, we saw in the proof of (\ref{est2}) above that
    \begin{equation*}
        |R^X_{s,t}| \lesssim \pvar{X}{p}{J}^2 + \pvar{\gamma}{\frac{p}{2}}{J}
    \end{equation*} 
    so we may drop $\pvar{R^X}{\frac{p}{2}}{[s, t]}$ from the right-hand side in (\ref{progress}). Supposing without loss of generality that $\pvar{\zeta}{p}{[s, t]}, \pvar{\lift{\zeta}{}}{\frac{p}{2}}{[s, t]} \le \frac{1}{2}$ whenever $|t - s| < \delta$ (we can do this and extend to $J = [0, T]$ via \textbf{Lemma~\ref{lemma:inequality2}}), we have
    \begin{equation}
    \label{remainderEst4}
    \begin{split}
        \pvar{R^X}{\frac{p}{2}}{[s, t]} & \lesssim \left( \pvar{X}{p}{[s, t]}^2 + \pvar{\gamma}{\frac{p}{2}}{[s, t]} \right) \pvar{\zeta}{p}{[s, t]} \\
        & + \left( 1 + \pvar{X}{p}{[s, t]}^2 + \pvar{\gamma}{\frac{p}{2}}{[s, t]} \right) \pvar{\lift{\zeta}{}}{\frac{p}{2}}{[s, t]} \\
        & + |t - s| + \pvar{\lift{\zeta}{}}{\frac{p}{2}}{[s, t]} \\
    \end{split}
    \end{equation}
    and by definition of $\pvar{R^X}{\frac{p}{2}}{[s, t]}$ \[\pvar{X}{p}{[s, t]} \lesssim \pvar{\zeta}{p}{[s, t]} + \pvar{R^X}{\frac{p}{2}}{[s, t]}\] hence
    \begin{equation*}
    \begin{split}
        \pvar{X}{p}{[s, t]} & \lesssim \pvar{\zeta}{p}{[s, t]} + \pvar{R^X}{\frac{p}{2}}{[s, t]} \\
        & \lesssim \pvar{\zeta}{p}{[s, t]} + \left( \pvar{X}{p}{[s, t]}^2 + \pvar{\gamma}{\frac{p}{2}}{[s, t]} \right) \pvar{\zeta}{p}{[s, t]} \\
        & + \left( 1 + \pvar{X}{p}{[s, t]}^2 + \pvar{\gamma}{\frac{p}{2}}{[s, t]} \right) \pvar{\lift{\zeta}{}}{\frac{p}{2}}{[s, t]} \\
        & + |t - s| + \pvar{\lift{\zeta}{}}{\frac{p}{2}}{[s, t]}. \\
    \end{split}
    \end{equation*}
    Expanding the right-hand side and noting that $\pvar{X}{p}{[s, t]}^{2} \pvar{\zeta}{p}{[s, t]} \le \pvar{X}{p}{[s, t]}^{2}$ since $\pvar{\zeta}{p}{[s, t]} < \frac{1}{2}$, we get
    \begin{equation*}
        \pvar{X}{p}{[s, t]} \lesssim \left( 1 + \pvar{\gamma}{\frac{p}{2}}{[s, t]} \right) \left( \pvar{\zeta}{p}{[s, t]} + \pvar{\lift{\zeta}{}}{\frac{p}{2}}{[s, t]} + |t - s|\right) + \pvar{X}{p}{[s, t]}^2.
    \end{equation*}
    Let $I \subseteq [s, t]$ denote a sub-interval such that $\pvar{X}{p}{I} < \frac{1}{2}$ and set $r$ equal to the length of $I$. Then
    \begin{equation*}
    \begin{split}
        \pvar{X}{p}{I} & \lesssim \left( 1 + \pvar{\gamma}{\frac{p}{2}}{I} \right) \left( \pvar{\zeta}{p}{I} + \pvar{\lift{\zeta}{}}{\frac{p}{2}}{I} + |I|\right) + \pvar{X}{p}{I}^2 \\
        & \lesssim 2 \left( 1 + \pvar{\gamma}{\frac{p}{2}}{I} \right) \left( \pvar{\zeta}{p}{I} + \pvar{\lift{\zeta}{}}{\frac{p}{2}}{I} + |I|\right) \\
        & \lesssim \left( 1 + \pvar{\gamma}{\frac{p}{2}}{I} \right) \left( M (\delta^*)^\frac{1}{p} + M (\delta^*)^\frac{2}{p} + \delta^* \right) \\
        & \lesssim 1 + \pvar{\gamma}{\frac{p}{2}}{I}. \\
    \end{split}
    \end{equation*}
    where $|I| = r$. Now we extend to $J = [0, T]$. Set $\delta^* = \min \{\delta, r\}$ and choose a partition of $J$ as in \textbf{Lemma~\ref{lemma:inequality2}} such that the mesh size that is lesser than $\delta^*$. Then
    \begin{equation*}
        \left( 1 + \pvar{\gamma}{\frac{p}{2}}{I} \right)^p \le 2^p  \left( 1 + \pvar{\gamma}{\frac{p}{2}}{I}^p \right)
    \end{equation*}
    by \textbf{Lemma~\ref{lemma:inequality1}}, so
    \begin{equation*}
    \begin{split}
        \pvar{X}{p}{J} & \lesssim \sum_{\mathcal{D}} \left( 1 + \pvar{\gamma}{\frac{p}{2}}{[t_i, t_{i+1}]}^p \right) \\
        & \lesssim 1 + \pvar{\gamma}{\frac{p}{2}}{J}^p
    \end{split}
    \end{equation*}
    but $1 + \pvar{\gamma}{\frac{p}{2}}{J}^p \lesssim 1 + \pvar{\gamma}{\frac{p}{2}}{J}^{1 + p}$, proving (\ref{est3}).

    To prove (\ref{est4}) we use (\ref{remainderEst4}) and approach the situation in an analogous fashion, arriving at the inequality
    \begin{equation*}
    \begin{split}
        \pvar{R^X}{\frac{p}{2}}{J} & \lesssim 1 + \pvar{\gamma}{\frac{p}{2}}{J}^{p/2} \\
        & \lesssim 1 + \pvar{\gamma}{\frac{p}{2}}{J}^{2+p}. \\
    \end{split}
    \end{equation*}
\end{proof}

\begin{proof}[Proof of Lemma \ref{lemma:indicator}]
\normalsize
    Since $\{\tau \le t\}$ is $\bb{F}_t$-measurable, there exists a Borel set $A_t$ such that \[ (\Hat{\bb{X}}|_{[0,t]})^{-1}(A_t) = \{\tau \le t\}. \] Hence, $\mathds{1}_{\{\tau \le t\}} = \mathds{1}_{A_t}(\Hat{\bb{X}}|_{[0,t]})$.

    Define the map 
    \[
    \phi \colon \Lambda_T^p \to [0,T] \times \Hat{\Omega}_T^p,\quad (t, \Hat{\bb{X}}|_{[0,t]}) \mapsto (t, \Tilde{\Hat{\bb{X}}}|_{[0,T]}),
    \]
    where $(\Tilde{\Hat{\bb{X}}}|_{[0,T]})_s = (s, (s \land t, \Hat{\bb{X}}|_{[0, s \land t]}))$.

    Now define
    \[
    f \colon [0, T] \times \Hat{\Omega}_T^p \to \bb{R},\quad (t, \Hat{\bb{X}}) \mapsto \mathds{1}_{A_t}(\Hat{\bb{X}}|_{[0,t]}).
    \]
    Let $I_k^n = \left[ \frac{k}{2^n}T, \frac{k+1}{2^n}T \right]$ for $k = 0, \dots, 2^n - 2$ and $I_{2^n-1}^n = \left[ \frac{2^n - 1}{2^n}T, T \right]$. Set $t_k^n = \frac{k}{2^n}T$.

    Define
    \[
    f_n(t, \Hat{\bb{X}}) = \sum_{k=0}^{2^n - 1} f(t_k^n, \Hat{\bb{X}})\, \mathds{1}_{I_k^n}(t),
    \]
    \[
    \Tilde{f}(t, \Hat{\bb{X}}) = \limsup_{m \to \infty} \limsup_{n \to \infty} f_n(t + m^{-1}, \Hat{\bb{X}}),
    \]
    and set
    \[
    \theta(\Hat{\bb{X}}|_{[0,t]}) = (\Tilde{f} \circ \phi)(\Hat{\bb{X}}|_{[0,t]}),
    \]
    so that
    \begin{equation*}
    \begin{split}
        \theta(\Hat{\bb{X}}|_{[0,t]}) & = \limsup_{m \to \infty} \limsup_{n \to \infty} \sum_{k=0}^{2^n - 1} \mathds{1}_{\{\tau \le t_k^n\}} \mathds{1}_{I_k^n}(t + m^{-1}) \\
        & = \mathds{1}_{\{\tau \le t\}}.
    \end{split}
    \end{equation*}
\end{proof}

\begin{proof}[Proof of Lemma \ref{lemma:exponential_differential_equation}]
\normalsize
    Since
    \begin{equation*}
    \begin{split}
        \frac{d}{dt} \langle w1, \Hat{\bb{X}}^{< \infty}_{0,t} \rangle & = \frac{d}{dt} \int_0^t \langle w1, \Hat{\bb{X}}^{< \infty}_{0,s} \rangle ds \\
        & = \langle w, \Hat{\bb{X}}^{< \infty}_{0,t} \rangle \\
    \end{split}
    \end{equation*}
    and $\langle l1, \mathbf{1} \rangle = 0$, we have by \textbf{Lemma \ref{lemma:exponential_shuffle_inequality}}
    \begin{equation*}
    \begin{split}
        & \frac{d}{dt}\langle \exp^\shuffle (l1), \Hat{\bb{X}}^{\le N}_{0,t} \rangle \\ 
        & = \sum_{0 \le k_1 deg(w_11) + \cdots + k_n deg(w_n1) \le N} \frac{\langle \lambda_1 w_1 1, \Hat{\bb{X}}^{< \infty}_{0,t} \rangle^{k_1}}{k_1!} \cdots \frac{\langle \lambda_n w_n 1, \Hat{\bb{X}}^{< \infty}_{0,t} \rangle^{k_n}}{k_n!} \\
        & = \sum_{0 \le k_1 deg(w_11) + \cdots + k_n deg(w_n1) \le N} \langle \lambda_1 w_1, \Hat{\bb{X}}^{< \infty}_{0,t} \rangle^{k_1} \frac{\langle \lambda_1 w_1 1, \Hat{\bb{X}}^{< \infty}_{0,t} \rangle^{k_1-1}}{(k_1-1)!} \cdots \frac{\langle \lambda_n w_n 1, \Hat{\bb{X}}^{< \infty}_{0,t} \rangle^{k_n}}{k_n!} \\
        & = \sum_{0 \le k_1 deg(w_11) + \cdots + k_n deg(w_n1) \le N} \Bigg[ \langle \lambda_n w_n, \Hat{\bb{X}}^{< \infty}_{0,t} \rangle^{k_n} \frac{\langle \lambda_1 w_1 1, \Hat{\bb{X}}^{< \infty}_{0,t} \rangle^{k_1}}{k_1!} \cdots \\
        & \cdots \frac{\langle \lambda_{n-1} w_{n-1} 1, \Hat{\bb{X}}^{< \infty}_{0,t} \rangle^{k_{n-1}}}{k_{n-1}!} \frac{\langle \lambda_n w_n 1, \Hat{\bb{X}}^{< \infty}_{0,t} \rangle^{k_n-1}}{(k_n-1)!} \Bigg]. \\
    \end{split}
    \end{equation*}
    But
    \begin{equation*}
    \begin{split}
        & \sum_{0 \le k_1 deg(w_11) + \cdots + k_n deg(w_n1) \le N} \frac{\langle \lambda_1 w_1 1, \Hat{\bb{X}}^{< \infty}_{0,t} \rangle^{{k_1-1}}}{(k_1-1)!} \cdots \frac{\langle \lambda_n w_n 1, \Hat{\bb{X}}^{< \infty}_{0,t} \rangle^{k_n}}{k_n!} \\
        & = \sum_{0 \le (k_1+1) deg(w_11) + \cdots + k_n deg(w_n1) \le N} \frac{\langle \lambda_1 w_1 1, \Hat{\bb{X}}^{< \infty}_{0,t} \rangle^{k_1}}{k_1!} \cdots \frac{\langle \lambda_n w_n 1, \Hat{\bb{X}}^{< \infty}_{0,t} \rangle^{k_n}}{k_n!} \\
        & = \langle \exp^\shuffle (l1), \Hat{\bb{X}}^{\le N - \text{deg}(w_1)-1}_{0,t} \rangle. \\
    \end{split}
    \end{equation*}
    Noting this and applying it to the remaining indices proves the result.
\end{proof}

\begin{proof}[Proof of Proposition \ref{proposition:integral_representation}]
\normalsize
    Since $\tau_\theta^r \in [0, T] \cup \{\infty\}$, we have
    \begin{equation*}
    \begin{split}
        P[\tau_\theta^r \le t \mid \Hat{\bb{X}}] &= P\left[ \int_0^{t \land T} \theta(\Hat{\bb{X}}|_{[0,s]})^2 \, ds \ge Z \mid \Hat{\bb{X}} \right] \\
        &= F_Z \left( \int_0^{t \land T} \theta(\Hat{\bb{X}}|_{[0,s]})^2 \, ds \right) \\
        &= \Tilde{F}(t).
    \end{split}
    \end{equation*}
    Also,
    \[
        P[\tau_\theta^r = \infty \mid \Hat{\bb{X}}] = P\left[ \int_0^T \theta(\Hat{\bb{X}}|_{[0,s]})^2 \, ds < Z \mid \Hat{\bb{X}} \right] = 1 - \Tilde{F}(T).
    \]
    Thus, for any integrable function $f \colon [0, \infty] \to \bb{R}$, we have
    \[
        \bb{E}[f(\tau_\theta^r) \mid \Hat{\bb{X}}] = \int_0^T f(t) \, d\Tilde{F}(t) + f(\infty)(1 - \Tilde{F}(T)).
    \]
    Applying this to $f(t) = Y_{t \land S}$ yields
    \begin{equation*}
    \begin{split}
        \bb{E}[Y_{\tau_\theta^r \land S} \mid \Hat{\bb{X}}] &= \int_0^T Y_{t \land S} \, d\Tilde{F}(t) + Y_S(1 - \Tilde{F}(T)) \\
        &= \int_0^S Y_t \, d\Tilde{F}(t) + Y_S \int_S^T d\Tilde{F}(t) + Y_S(1 - \Tilde{F}(T)) \\
        &= \int_0^S Y_t \, d\Tilde{F}(t) + Y_S(1 - \Tilde{F}(S)).
    \end{split}
    \end{equation*}
    Using integration by parts and the fact that $\Tilde{F}(0) = 0$, we get
    \begin{equation*}
    \begin{split}
        \int_0^S Y_t \, d\Tilde{F}(t) &= Y_S \Tilde{F}(S) - \int_0^S \Tilde{F}(t) \, dY_t, \\
        \bb{E}[Y_{\tau_\theta^r \land S} \mid \Hat{\bb{X}}] &= Y_S - \int_0^S \Tilde{F}(t) \, dY_t = \int_0^S (1 - \Tilde{F}(t)) \, dY_t + Y_0.
    \end{split}
    \end{equation*}
\end{proof}

\end{appendices}

\bibliography{sn-bibliography}%


\begin{thebibliography}{14}
\ifx \bisbn   \undefined \def \bisbn  #1{ISBN #1}\fi
\ifx \binits  \undefined \def \binits#1{#1}\fi
\ifx \bauthor  \undefined \def \bauthor#1{#1}\fi
\ifx \batitle  \undefined \def \batitle#1{#1}\fi
\ifx \bjtitle  \undefined \def \bjtitle#1{#1}\fi
\ifx \bvolume  \undefined \def \bvolume#1{\textbf{#1}}\fi
\ifx \byear  \undefined \def \byear#1{#1}\fi
\ifx \bissue  \undefined \def \bissue#1{#1}\fi
\ifx \bfpage  \undefined \def \bfpage#1{#1}\fi
\ifx \blpage  \undefined \def \blpage #1{#1}\fi
\ifx \burl  \undefined \def \burl#1{\textsf{#1}}\fi
\ifx \doiurl  \undefined \def \doiurl#1{\url{https://doi.org/#1}}\fi
\ifx \betal  \undefined \def \betal{\textit{et al.}}\fi
\ifx \binstitute  \undefined \def \binstitute#1{#1}\fi
\ifx \binstitutionaled  \undefined \def \binstitutionaled#1{#1}\fi
\ifx \bctitle  \undefined \def \bctitle#1{#1}\fi
\ifx \beditor  \undefined \def \beditor#1{#1}\fi
\ifx \bpublisher  \undefined \def \bpublisher#1{#1}\fi
\ifx \bbtitle  \undefined \def \bbtitle#1{#1}\fi
\ifx \bedition  \undefined \def \bedition#1{#1}\fi
\ifx \bseriesno  \undefined \def \bseriesno#1{#1}\fi
\ifx \blocation  \undefined \def \blocation#1{#1}\fi
\ifx \bsertitle  \undefined \def \bsertitle#1{#1}\fi
\ifx \bsnm \undefined \def \bsnm#1{#1}\fi
\ifx \bsuffix \undefined \def \bsuffix#1{#1}\fi
\ifx \bparticle \undefined \def \bparticle#1{#1}\fi
\ifx \barticle \undefined \def \barticle#1{#1}\fi
\bibcommenthead
\ifx \bconfdate \undefined \def \bconfdate #1{#1}\fi
\ifx \botherref \undefined \def \botherref #1{#1}\fi
\ifx \url \undefined \def \url#1{\textsf{#1}}\fi
\ifx \bchapter \undefined \def \bchapter#1{#1}\fi
\ifx \bbook \undefined \def \bbook#1{#1}\fi
\ifx \bcomment \undefined \def \bcomment#1{#1}\fi
\ifx \oauthor \undefined \def \oauthor#1{#1}\fi
\ifx \citeauthoryear \undefined \def \citeauthoryear#1{#1}\fi
\ifx \endbibitem  \undefined \def \endbibitem {}\fi
\ifx \bconflocation  \undefined \def \bconflocation#1{#1}\fi
\ifx \arxivurl  \undefined \def \arxivurl#1{\textsf{#1}}\fi
\csname PreBibitemsHook\endcsname

\bibitem[\protect\citeauthoryear{Friz and Hairer}{2020}]{frizHairer}
\begin{bbook}
\bauthor{\bsnm{Friz}, \binits{P.K.}},
\bauthor{\bsnm{Hairer}, \binits{M.}}:
\bbtitle{A Course on Rough Paths}.
\bpublisher{Springer},
\blocation{Switzerland}
(\byear{2020})
\end{bbook}
\endbibitem

\bibitem[\protect\citeauthoryear{Lyons et~al.}{2007}]{lyons}
\begin{bbook}
\bauthor{\bsnm{Lyons}, \binits{T.J.}},
\bauthor{\bsnm{Caruana}, \binits{M.}},
\bauthor{\bsnm{L{\'e}vy}, \binits{T.}}:
\bbtitle{Differential Equations Driven by Rough Paths}.
\bsertitle{Lecture Notes in Mathematics},
vol. \bseriesno{1908}.
\bpublisher{Springer},
\blocation{Heidelberg, Germany}
(\byear{2007})
\end{bbook}
\endbibitem

\bibitem[\protect\citeauthoryear{Gubinelli}{2004}]{gubinelliControlledRoughPath}
\begin{barticle}
\bauthor{\bsnm{Gubinelli}, \binits{M.}}:
\batitle{Controlling rough paths}.
\bjtitle{Journal of Functional Analysis}
\bvolume{216}(\bissue{1}),
\bfpage{86}--\blpage{140}
(\byear{2004})
\doiurl{10.1016/j.jfa.2004.01.002}
\end{barticle}
\endbibitem

\bibitem[\protect\citeauthoryear{Diehl et~al.}{2017}]{DiehlFrizGassiat}
\begin{barticle}
\bauthor{\bsnm{Diehl}, \binits{J.}},
\bauthor{\bsnm{Friz}, \binits{P.K.}},
\bauthor{\bsnm{Gassiat}, \binits{P.}}:
\batitle{Stochastic control with rough paths}.
\bjtitle{Applied Mathematics \& Optimization}
\bvolume{75},
\bfpage{285}--\blpage{315}
(\byear{2017})
\end{barticle}
\endbibitem

\bibitem[\protect\citeauthoryear{Allan and Cohen}{2019}]{allan2019pathwise}
\begin{botherref}
\oauthor{\bsnm{Allan}, \binits{A.L.}},
\oauthor{\bsnm{Cohen}, \binits{S.N.}}:
Pathwise Stochastic Control with Applications to Robust Filtering
(2019)
\end{botherref}
\endbibitem

\bibitem[\protect\citeauthoryear{Loomis and Sternberg}{2014}]{sternberg}
\begin{bbook}
\bauthor{\bsnm{Loomis}, \binits{L.H.}},
\bauthor{\bsnm{Sternberg}, \binits{S.}}:
\bbtitle{Advanced Calculus},
\bedition{Revised edition} edn.
\bpublisher{World Scientific Publishing Co.},
\blocation{Singapore}
(\byear{2014})
\end{bbook}
\endbibitem

\bibitem[\protect\citeauthoryear{Yong and Zhou}{1999}]{yong}
\begin{bbook}
\bauthor{\bsnm{Yong}, \binits{J.}},
\bauthor{\bsnm{Zhou}, \binits{X.Y.}}:
\bbtitle{Stochastic Controls: Hamiltonian Systems and HJB Equations}
vol. \bseriesno{43}.
\bpublisher{Springer},
\blocation{New York}
(\byear{1999})
\end{bbook}
\endbibitem

\bibitem[\protect\citeauthoryear{Bardi and Da~Lio}{1997}]{bardi}
\begin{barticle}
\bauthor{\bsnm{Bardi}, \binits{M.}},
\bauthor{\bsnm{Da~Lio}, \binits{F.}}:
\batitle{On the bellman equation for some unbounded control problems}.
\bjtitle{Nonlinear Differential Equations and Applications NoDEA}
\bvolume{4},
\bfpage{491}--\blpage{510}
(\byear{1997})
\end{barticle}
\endbibitem

\bibitem[\protect\citeauthoryear{Bain and Crisan}{2009}]{bain}
\begin{bbook}
\bauthor{\bsnm{Bain}, \binits{A.}},
\bauthor{\bsnm{Crisan}, \binits{D.}}:
\bbtitle{Fundamentals of Stochastic Filtering}
vol. \bseriesno{3}.
\bpublisher{Springer},
\blocation{Heidelberg, Germany}
(\byear{2009})
\end{bbook}
\endbibitem

\bibitem[\protect\citeauthoryear{Øksendal}{2013}]{oksendal}
\begin{bbook}
\bauthor{\bsnm{Øksendal}, \binits{B.}}:
\bbtitle{Stochastic Differential Equations: An Introduction with Applications},
\bedition{6}th edn.
\bpublisher{Springer},
\blocation{New York}
(\byear{2013})
\end{bbook}
\endbibitem

\bibitem[\protect\citeauthoryear{Allan and Cohen}{2019}]{ac}
\begin{barticle}
\bauthor{\bsnm{Allan}, \binits{A.L.}},
\bauthor{\bsnm{Cohen}, \binits{S.N.}}:
\batitle{Parameter uncertainty in the kalman--bucy filter}.
\bjtitle{SIAM Journal on Control and Optimization}
\bvolume{57}(\bissue{3}),
\bfpage{1646}--\blpage{1671}
(\byear{2019})
\end{barticle}
\endbibitem

\bibitem[\protect\citeauthoryear{Wi{\'s}niewski}{1994}]{wisn}
\begin{barticle}
\bauthor{\bsnm{Wi{\'s}niewski}, \binits{A.}}:
\batitle{The structure of measurable mappings on metric spaces}.
\bjtitle{Proceedings of the American Mathematical Society}
\bvolume{122}(\bissue{1}),
\bfpage{147}--\blpage{150}
(\byear{1994})
\end{barticle}
\endbibitem

\bibitem[\protect\citeauthoryear{Kalsi et~al.}{2020}]{kalsi}
\begin{barticle}
\bauthor{\bsnm{Kalsi}, \binits{J.}},
\bauthor{\bsnm{Lyons}, \binits{T.}},
\bauthor{\bsnm{Arribas}, \binits{I.P.}}:
\batitle{Optimal execution with rough path signatures}.
\bjtitle{SIAM Journal on Financial Mathematics}
\bvolume{11}(\bissue{2}),
\bfpage{470}--\blpage{493}
(\byear{2020})
\end{barticle}
\endbibitem

\bibitem[\protect\citeauthoryear{Bayer et~al.}{2023}]{bayer2021optimalstoppingsignatures}
\begin{barticle}
\bauthor{\bsnm{Bayer}, \binits{C.}},
\bauthor{\bsnm{Hager}, \binits{P.P.}},
\bauthor{\bsnm{Riedel}, \binits{S.}},
\bauthor{\bsnm{Schoenmakers}, \binits{J.}}:
\batitle{Optimal stopping with signatures}.
\bjtitle{The Annals of Applied Probability}
\bvolume{33}(\bissue{1}),
\bfpage{238}--\blpage{273}
(\byear{2023})
\end{barticle}
\endbibitem

\end{thebibliography}

\end{document}